\documentclass[conference]{IEEEtran}

\usepackage{url}

\usepackage{times}
\usepackage{helvet}
\usepackage{courier}
\usepackage{xcolor}
\usepackage{float}
\usepackage{graphicx}
\usepackage{soul}
\usepackage[utf8]{inputenc}

\usepackage[ruled,vlined,linesnumbered]{algorithm2e}
\usepackage{amssymb, amsmath}
\usepackage{amsthm}

\usepackage{booktabs}
\usepackage{graphicx,multirow,makecell}
\usepackage{subfigure}
\graphicspath{{fig/}}

\newcommand{\AlgNWF}{{\textsc{Greedy}}\xspace}
\newcommand{\AlgRDP}{{\textsc{RDPGreed}}\xspace}
\newcommand{\AlgPG}{{\textsc{PresGreed}}\xspace}
\newcommand{\AlgSPG}{{\textsc{StocPresGreed}}\xspace}

\newcommand{\AlgSPH}{{\textsc{Sphere}}\xspace}
\newcommand{\AlgHS}{{\textsc{RMS-HS}}\xspace}
\newcommand{\AlgHDR}{{\textsc{HD-RRMS}}\xspace}
\newcommand{\AlgHDG}{{\textsc{HD-Greedy}}\xspace}
\newcommand{\AlgKER}{{\textsc{Eps-Kernel}}\xspace}

\newtheorem{theorem}{Theorem}
\newtheorem{lemma}{Lemma}

\newtheorem{example}{Example}
\newtheorem{heuristic}{Heuristic}

\DeclareMathOperator*{\argmax}{arg\,max}
\newcommand{\myparagraphnew}[1]{\vspace{1mm} \noindent \textbf{#1}}

\newcommand{\ie}{\emph{i.e.,}\xspace}
\newcommand{\eg}{\emph{e.g.,}\xspace}
\newcommand{\myparagraph}[1]{\vspace{1mm} \noindent \textbf{#1}.}

\let\oldnl\nl
\newcommand{\nonl}{\renewcommand{\nl}{\let\nl\oldnl}}

\usepackage{soul}

\begin{document}

\title{K-Regret Minimization Acceleration: A Sampling-based Approach}
\title{Efficient Processing of $k$-regret Minimization Queries with Theoretical Guarantees}

\author{
{Jiping Zheng{\small$^{\#}$}, Qi Dong{\small $^{\#}$}, Xiaoyang Wang{\small $^{\dag}$}, Ying Zhang{\small $^{\ddag}$}, Wei Ma{\small $^{\#}$}, Yuan Ma{\small $^{\#}$}}
\vspace{1.6mm}\\
\fontsize{10}{10}\selectfont\itshape
$~^{\#}$College of Computer Science and Technology, Nanjing University of Aeronautics and Astronautics, Nanjing, China\\
\fontsize{9}{9}\selectfont\ttfamily\upshape

\fontsize{10}{10}\selectfont\rmfamily\itshape
$~^{\dag}$School of Computer Science and Information Engineering, Zhejiang Gongshang University, Hangzhou, China\\
$~^{\ddag}$School of Computer Science, University of Technology Sydney, Sydney, Australia\\

Email: \{jzh, dongqi, mawei, mayuancs\}@nuaa.edu.cn, xiaoyangw@zjgsu.edu.cn, ying.zhang@uts.edu.au}

\maketitle

\begin{abstract}
	Assisting end users to identify desired results from a large dataset is an important problem for multi-criteria decision making. To address this problem, top-$k$ and skyline queries have been widely adopted, but they both have inherent drawbacks, \ie the user either has to provide a specific utility function or faces many results. The $k$-regret minimization query is proposed, which integrates the merits of top-$k$ and skyline queries. Due to the NP-hardness of the problem, the $k$-regret minimization query is time consuming and the greedy framework is widely adopted. However, formal theoretical analysis of the greedy approaches for the quality of the returned results is still lacking. In this paper, we first fill this gap by conducting a nontrivial theoretical analysis of the approximation ratio of the returned results. To speed up query processing, a sampling-based method, \AlgSPG, is developed to reduce the evaluation cost. In addition, a theoretical analysis of the required sample size is conducted to bound the quality of the returned results. Finally, comprehensive experiments are conducted on both real and synthetic datasets to demonstrate the efficiency and effectiveness of the proposed methods.
\end{abstract}

\section{Introduction}

For end users, identifying the most desired data points from a large dataset to support multi-criteria decision making is an essential functionality in many domains.
For instance, when online shopping, it is difficult for a user to view all products. A possible solution is to show the user some representative products based
on certain criteria, \eg the user's preferences. In the literature,
top-$k$ queries \cite{Ilyas:2008} and skyline queries \cite{Borzsony:2001} are two well-studied and commonly used tools that can effectively reduce the output size.
However, both queries suffer from some drawbacks. In the top-$k$ query, we assume that the user can provide a utility function (\ie user's preference function) in advance. Then, it can output $k$ data points with the largest scores. In real applications, users usually do not have a specific function in mind.
A skyline query does not ask the user for any utility function.
Instead, it retrieves all the data points that are not dominated by others in the dataset.
For instance, Table~\ref{tab:nbaexample} shows a skyline query toy example. It consists of the 16 skyline NBA players found from the 2009 regular season by considering 3 player attributes, \ie points, rebounds and steals\footnote{\url{https://www.rotowire.com/basketball/}}. These players do not \emph{dominate} each other, such as Kevin Durant and LeBron James, where
Kevin Durant is better in points and rebounds than LeBron James but not in steals.
However, in real scenarios, the exact number of skyline results is generally large and uncontrollable and cannot be foreseen before the whole dataset is accessed. In addition, the output size of skyline queries usually increases rapidly with the dimensionality. Consequently, the user will still face many choices.

\begin{table}[htb]
	\centering
	\caption{Skyline NBA players from the 2009 regular season.}
	\label{NBA}
	\footnotesize
	\renewcommand{\arraystretch}{1.2}	
	\begin{tabular}{|c|l|c|c|c|}
		\hline
		ID & Player Name & Points & Rebounds & Steals \\
		\hline\hline
		1 &  Kevin	Durant & 2,472 & 623 & 112 \\
		2 &  LeBron	James & 2,258 & 554 & 125 \\
		3 &  Dwyane	Wade & 2,045 & 373 & 142 \\
		4 &  Amare	Stoudemire & 1,896 & 732 & 52 \\
		5 &  Zach	Randolph & 1,681 & 950 & 80 \\
		6 &  Stephen	Jackson & 1,667 & 401 & 132 \\
		7 &  David	Lee & 1,640 & 949 & 85 \\
		8 &  Monta	Ellis & 1,631 & 257 & 143 \\
		9 &  Dwight	Howard & 1,503 & 1,082 & 75 \\
		10 & Andre	Iguodala & 1,401 & 529 & 141 \\
		11 & Stephen	Curry & 1,399 & 356 & 152 \\
		12 & Gerald	Wallace & 1,386 & 762 & 117 \\
		13 & Josh	Smith & 1,269 & 705 & 130 \\
		14 & Rajon	Rondo & 1,110 & 360 & 189 \\
		15 & Jason	Kidd & 824 & 445 & 145 \\	
		16 & Marcus	Camby & 556 & 871 & 95 \\
		\hline
	\end{tabular}
	\label{tab:nbaexample}
\end{table}

To address these problems, Nanongkai et al. \cite{Nanongkai:2010} introduced the $k$-regret minimization query. Given a positive integer $k$ and a dataset $D$, it returns a set $S$ of $k$ data points that can minimize the user's maximum regret ratio under a class of utility functions. 
The $k$-regret minimization query integrates the merits of top-$k$ and skyline queries. 
It does not ask the user for any utility function and outputs a result with controllable size, \ie only $k$ data points.
Due to the NP-hardness of the problem, in \cite{Nanongkai:2010}, \AlgRDP was proposed based on the Ramer-Douglas-Peucker greedy framework. 
The greedy strategy has also been used in most follow-up studies, \eg~\cite{Peng:2014,Faulkner:2015,Xie:2018,Qi:2018,Qiu:2018,Dong:2019}.

To answer the $k$-regret query, we need to exploit the linear utility function space to simulate all possible utility functions that users may be interested in~\cite{Xie:2019VLDBJ}. However, it has a side effect of performing many function evaluations via linear programming (LP). Generally, the number of function evaluations conducted by \AlgRDP is $nk$, where $n$ is the cardinality of the whole dataset.
In real applications, $k$ is usually much smaller than the size of the skyline set. Thus, the skyline points are considered candidate points, whose setting has also been widely adopted by follow-up studies.
Even if $n$ is the cardinality of the skyline, the cost of LPs still significantly affects the efficiency of the $k$-regret minimization query.
Therefore, due to the time complexity, \AlgRDP cannot scale well for larger $n$ and $k$.

Given the importance of the $k$-regret query, the following studies attempt to improve the performance from different perspectives. In~\cite{Peng:2014}, the authors attempted to address the scalability issues by reducing the candidate size, but the time complexity of decreasing the candidate set is still unacceptable when $n$ is large. In \cite{Cao:2017,Agarwal:2017,Asudeh:2017}, the proposed approaches, such as $\epsilon$-kernel, hitting set and \AlgHDG, output the results of variable sizes. Some studies attempted to accelerate the processing by sacrificing the quality of returned results or compromising between the size of the returned set and the maximum regret ratio bounds, \eg~\cite{Asudeh:2017,Agarwal:2017,Cao:2017,Xie:2018}.
However, except for \cite{Xie:2018}, which partially adopts the greedy framework, none of them has attempted to speed up or optimize the greedy framework in \AlgRDP~\cite{Nanongkai:2010}. 
Moreover, there is still no theoretical analysis of the greedy framework, which is widely adopted by the existing research.

In this paper, the main objective is to speed up existing greedy algorithms with theoretical guarantees.
The developed techniques and theorems can also be used by the existing research if a similar greedy framework is involved.
To simplify the theoretical analysis, we adopt the concept of the \emph{happiness ratio} \cite{Xie:2020}, which captures how happy a user can be after seeing $k$ representative data points instead of the whole dataset. We show that minimizing the maximum regret ratio is equal to maximizing the user's minimum happiness ratio. The minimum happiness ratio function is a monotone set function, which is composed of a class of submodular functions corresponding to a class of utility functions of the regret minimization query. We know that the function for computing the maximum of a class of submodular functions is generally not submodular, which means that the minimum happiness ratio function is also not submodular. Therefore, executing the $k$-regret minimization query can be regarded as a monotone nonsubmodular maximization problem with cardinality constraint (\ie return $k$ data points).
By adopting the happiness ratio in \AlgRDP, we introduce the \AlgPG algorithm and conduct a rigorous theoretical analysis of the approximation ratio. Moreover, to further speed up the computation, we develop a sampling-based approach, \AlgSPG, to achieve a good tradeoff between query processing time and result quality, also with theoretical guarantees.

\myparagraph{Contributions} Our principal contributions are summarized as follows.
\begin{itemize}
	\item In this paper, 
	we conduct the first work to provide rigorous theoretical analysis about the approximation ratio for \AlgRDP.
	\item To reduce the number of evaluations to improve the efficiency, a sampling-based algorithm, \AlgSPG, is proposed.  We also provide theoretical guarantees about the sample size required to bound the quality of returned results.
	\item Finally, through comprehensive experiments over synthetic and real-world datasets, we demonstrate the efficiency and effectiveness of the proposed methods.
\end{itemize}

\myparagraph{Road map} The rest of the paper is organized as follows. We briefly introduce the problem to be studied and the related techniques in Section~\ref{sec:background}. In Section~\ref{sec:sampling}, we provide the theoretical analysis as well as the sampling-based method. We demonstrate the efficiency and effectiveness of the proposed framework in Section~\ref{sec:exp} on both synthetic and real-world datasets and introduce the related works in Section~\ref{sec:rel}. Finally, we conclude the paper in Section~\ref{sec:conc}.

\section{Preliminaries}
\label{sec:background}

In this section, we first formally introduce the
problem studied in this paper, as well as the properties of the objective function. Then
we provide some related concepts employed in the following sections.
Table~\ref{notation} summarizes the notations frequently used throughout
the paper.

\begin{table}[tb]
	\centering
	\caption{Notations}
	\label{notation}
	\footnotesize
	\begin{tabular}{|c | l |}\hline
		Symbol      &Description        \\ \hline\hline
		$D, S$      &the whole dataset, a subset of $D$      \\ \hline
		$k, n, d$  & a positive integer, cardinality of $D$, dimensionality of $D$  \\ \hline
		$u, f$  & utility function   \\ \hline
		$\mathcal{U}$  &a class of utility functions   \\ \hline
		$p, u(p)$  &a data point in $D$, the utility of $p$   \\ \hline
		$H_D(S, u)$  & happiness ratio of $S$ under a specific utility $u$\\ \hline
		$H_D (S,\mathcal{U})$  & minimum happiness ratio of $S$ \\ \hline
		$\Delta_{S}(p)$  & marginal gain of point $p$ respect to $S$ \\ \hline
		$S_i$  & the successive set in step $i$ of a greedy algorithm \\ \hline
		$q_i$  & the point selected into the solution set $S_{i-1}$ in step $i$ \\ \hline
		$\Delta_i$  & marginal gain in step $i$ of a greedy algorithm \\ \hline
		$\gamma, \gamma^G$  & submodularity ratio, greedy submodularity ratio \\ \hline
		$\alpha, \alpha^G$  & generalized curvature, greedy curvature \\ \hline
		$S^*$  & the optimal solution \\ \hline
		$o_i$  & the element in $S_k\cap S^*$ \\ \hline
		$I(\cdot), x_i$  & approximation ratio, the $i$th component of $I(\cdot)$ \\ \hline
		$\textbf{y}, y_i$  & the optimal solution of $I(\emptyset)$, the $i$th component of $\textbf{y}$ \\ \hline
		$\epsilon, \lambda, s$  & approximation accuracy, sample factor, sample size \\ \hline
	\end{tabular}
\end{table}

\subsection{K-Regret Minimization}

Let $D$ be a set of $n$ $d$-dimensional points with positive real values. For each point $p\in D$, the value on the $i$th dimension is represented as $p[i]$. Before we introduce the problem,
we first present some related concepts \cite{Nanongkai:2010}.

\myparagraph{Utility function} A utility function $u$ is a mapping $u$: $\mathbb{R}_+^d \rightarrow \mathbb{R}_+$.
Given a utility function $u$, the utility of a data point $p$ is denoted as $u(p)$, which
shows how satisfied the user is with the data point.

\myparagraph{Regret ratio} Given a dataset $D$, a subset $S$ of $D$ and a utility function $u$, the regret ratio of $S$, represented as $R_D(S, u)$, is defined as
\[
R_D(S,u)=1-\frac{\max_{p\in S} u(p)}{\max_{p\in D}u(p)}
\]	

Since $S$ is a subset of $D$, given a utility function $u$, it is obvious that $\max_{p\in S}u(p)\leq\max_{p\in D}u(p)$ and the regret ratio ranges from 0 to 1. The user with utility function $u$ will be happy if the regret ratio approaches 0 because the maximum utility of $S$ is close to the maximum utility of $D$.

\myparagraph{Linear utility function} Assume there are some nonnegative real values $\{v[1], v[2], \cdots, v[d]\}$, where $v[i]$ denotes the user's preference for the $i$th dimension. Then, a linear utility function can be represented by these nonnegative reals and $u(p)=\sum_{i=1}^d v[i]\cdot p[i]$. A linear utility function can also be expressed by a vector, \ie $v=<v[1], v[2], ..., v[d]>$, so the utility of point $p$ can be expressed by the dot product of $v$ and $p$, \ie $u(p)=v\cdot p$.


\myparagraph{Maximum regret ratio \cite{Nanongkai:2010}}
Given a dataset $D$, a subset $S$ of $D$ and a class of utility functions $\mathcal{U}$. The maximum regret ratio of $S$, represented as $R_D (S, \mathcal{U})$, is defined as 	
\[
R_D (S,\mathcal{U})=\sup_{u\in \mathcal{U}}R_D(S,u)=\sup_{u\in \mathcal{U}}\left(1-\frac{\max_{p\in S} u(p)}{\max_{p\in D} u(p)}\right)
\]

To calculate $R_D (S,\mathcal{U})$, we call an LP solver such as the \emph{Simplex} method 
or the \emph{Interior Point} method 
\cite{Bertsimas:1997}. To simplify the theoretical analysis, we introduce the concepts of \textbf{happiness ratio} $H_D(S,u)$ and \textbf{minimum happiness ratio} $H_D (S,\mathcal{U})$, where
\[
H_D(S,u)=1-R_D (S,u)
\]
\[
H_D (S,\mathcal{U})=1-R_D (S,\mathcal{U}) = \inf_{u\in \mathcal{U}}\frac{\max_{p\in S} u(p)}{\max_{p\in D} u(p)}
\]
\begin{table*}
	\centering
	\setlength{\tabcolsep}{0.5mm}
	\footnotesize
	\renewcommand{\arraystretch}{1.2}
	\caption{Example of our problem. We show utilities and happiness ratios of first 9 players in Table \ref{NBA} with the corresponding value normalized.}
	\label{NBA_hr}
	\begin{tabular}{|c|l|ccc||cc|cc|cc|cc|c|}
		\hline
		\multirow{2}{*}{id} & \multirow{2}{*}{player name} & \multicolumn{3}{c||}{normalized points} & \multicolumn{2}{c|}{$v_0=<.9,.05,.05>$} & \multicolumn{2}{c|}{$v_{1}=<.05,.9,.05>$} & \multicolumn{2}{c|}{$v_{2}<.05,.05,.9>$} & \multicolumn{2}{c|}{$v_{3}=<.33,.33,.34>$}  &   \multirow{2}{*}{$\min$ h.r.} \\ \cline{3-13}
		& & $x$ & $y$ & $z$ & utility & h.r. & utility & h.r. & utility & h.r. & utility & h.r. & \\
		\hline\hline
		1 &  Kevin	Durant & 1.00 & 0.58 & 0.59 & 0.96 & $\frac{0.96}{0.96}=1.00$ & 0.60 & 0.63 & 0.61 & 0.84 & 0.72 & 1.00 & 0.63\\
		2 &  LeBron	James & 0.91 & 0.51 & 0.66 & 0.88 & $\frac{0.88}{0.96}=0.92$ & 0.54 & 0.57 & 0.67 & 0.92 & 0.69 & 0.96 & 0.57 \\
		3 &  Dwyane	Wade & 0.83 & 0.34 & 0.75 & 0.80 & $\frac{0.80}{0.96}=0.83$ & 0.39 & 0.41 & 0.73 & 1.00 & 0.64 & 0.89 & 0.41\\
		4 &  Amare	Stoudemire & 0.77 & 0.68 & 0.28 & 0.74 & $\frac{0.74}{0.96}=0.77$ & 0.66 & 0.69 & 0.32 & 0.44 & 0.57 & 0.79 & 0.44 \\
		5 &  Zach	Randolph & 0.68 & 0.88 & 0.42 & 0.68 & $\frac{0.68}{0.96}=0.71$ & 0.85 & 0.89 & 0.46 & 0.63 & 0.66 & 0.92 & 0.63\\
		6 &  Stephen  Jackson & 0.67 & 0.37 & 0.70 & 0.66 & $\frac{0.66}{0.96}=0.69$ & 0.40 & 0.42 & 0.68 & 0.93 & 0.58 & 0.81 & 0.42\\
		7 &  David	Lee & 0.66 &0.88 & 0.45 & 0.66 & $\frac{0.66}{0.96}=0.69$ & 0.85 & 0.89 & 0.48 & 0.66 & 0.66 & 0.92 & 0.66\\
		8 &  Monta	Ellis & 0.66 &0.24 &0.76 & 0.64 & $\frac{0.64}{0.96}=0.67$ & 0.29 & 0.31 & 0.73 & 1.00 & 0.56 & 0.78 & 0.31\\
		9 &  Dwight	Howard & 0.61 &1.00 & 0.40 & 0.62 & $\frac{0.62}{0.96}=0.66$ & 0.95 & 1.00 & 0.44 & 0.60 & 0.67 & 0.93 & 0.60\\
		\hline
	\end{tabular}
\end{table*}

\begin{example}
	Table \ref{NBA_hr} shows utilities, happiness ratios and minimum happiness ratios of the first 9 skyline NBA players in Table~\ref{NBA} along with 4 linear utility functions
	$v_0=\left< 0.9,0.05,0.05 \right>$, $v_{1}=\left<0.05,0.9,0.05 \right>$, $v_{2}=\left<0.05,0.05,0.9\right>$ and $v_{3}=\left<0.33,0.33,0.34\right>$. For Durant, his corresponding utilities are 0.96, 0.60, 0.61 and 0.72 with happiness ratios of 1, 0.63, 0.84 and 1. The minimum happiness ratio of \{Durant\} is 0.63.
\end{example}

\myparagraph{Problem Definition} Given a dataset $D$, a positive integer $k$ and a class of linear utility functions $\mathcal{U}$, we attempt to efficiently answer the $k$-regret minimization query, which finds a subset $S^*$ of $D$ containing at most $k$ points such that the minimum happiness ratio is maximized, \ie
\[
S^* = \argmax_{S\subseteq D \wedge |S|\leq k} H_D (S,\mathcal{U})
\]

\subsection{Function Properties}

We first show that our happiness ratio function under a specified utility function is a monotone nondecreasing submodular function. 
Given a function $f:2^{D}\rightarrow \mathbb{R}$, $S\subseteq D$ and $p\in D\setminus S$, let $\Delta_{S}(p)\doteq f(\{p\}\cup S)-f(S)$ be the marginal gain of adding data point $p$ to the set $S$.
The function $f$ is monotonically nondecreasing and submodular if $i)$ $\Delta_{S}(p)\geq 0$ and $ii)$ for $S_{1}\subseteq S_{2}$ and $p \notin S_2$, $\Delta_{S_{1}}(p)\geq \Delta_{S_{2}}(p)$~\cite{Fujishige:2005}. An equivalent form of a submodular function is expressed by $f(S_{1})+f(S_{2})\geq f(S_{1}\cup S_{2})+f(S_{1}\cap S_{2})$ compared to a modular function $f$ which has the form $f(S_{1})+f(S_{2})= f(S_{1}\cup S_{2})+f(S_{1}\cap S_{2})$.

We observe that for a specific utility function $u$, our minimum happiness ratio function $H_D(S,u)$ is a monotonic nondecreasing submodular function. The monotonicity is obvious for $S_1\subseteq S_2$. The observation of submodularity is shown as follows. Let $\triangle'=H_D(S_1\cup \{q\},u)-H_D(S_1,u)$ and $\triangle''=H_D(S_2\cup \{q\},u)-H_D(S_2,u)$. We can see that the maximum value of $u(p)$ may be achieved in one of the three cases: $p\in S_1$, $p\in S_2 \backslash S_1$ or $p\in D\backslash S_2$. In all the three cases, $\triangle' \geq \triangle''$. Thus, $H_D(S,u)$ is submodular for a given utility function $u$. 

Unfortunately, submodularity does not hold for $H_{D}(S,\mathcal{U})$, where $\mathcal{U}$ represents a class of utility functions. We show that the claim is correct by constructing a counterexample. According to~\cite{Fujishige:2005}, a function $f$ is submodular if it satisfies $f(S_{1})+f(S_{2})\geq f(S_{1}\cup S_{2})+f(S_{1}\cap S_{2})$ for $\forall S_{1},S_{2}\subseteq D$. Given a dataset $D=\{(0,1),(1,0)\}$ and a function class $\mathcal{U}$ with two linear functions $\left< 1,0 \right>, \left<0,1 \right>$, let $S_1=\{(0,1)\}$ and $S_2=\{(1,0)\}$. Based on the definition, we have $H_D(D,\mathcal{U})=1$ and $H_D(S_1,\mathcal{U})=H_D(S_2,\mathcal{U})=H_D(\emptyset,\mathcal{U})=0$. It means $1= H_D(S_{1}\cup S_{2},\mathcal{U})+H_D(S_{1}\cap S_{2},\mathcal{U})\nleq H_D(S_{1},\mathcal{U})+H_D(S_{2},\mathcal{U}) = 0$. Therefore, the function $H_{D}(S,\mathcal{U})$ is not submodular.

Monotone nondecreasing submodular functions have been shown with an $1-1/e$ approximation ratio and can be well approximated, \ie near optimal \cite{Nemhauser:1978}.
Because $H_D(S,u)$ is submodular while $H_{D}(S,\mathcal{U})$ is not, we introduce the submodularity ratio and generalized curvature to model $H_{D}(S,\mathcal{U})$ to obtain the guarantees of the greedy framework for regret minimization (or happiness maximization) queries \cite{Peng:2014,Faulkner:2015,Xie:2018,Qi:2018,Qiu:2018,Dong:2019}.

\begin{lemma}
	\label{le:submo}
	For a specific utility function $u$, $H_D(S,u)$ is a monotone nondecreasing submodular function.
\end{lemma}

\begin{proof}
	\textit{Monotone nondecreasing:} Given $S_{1}$, $S_{2}$ where $S_{1}\subseteq S_{2}\subseteq D$ and a specific utility function $u$, the happiness ratio of $S_1$ (\ie $H_D(S_1,u)$) is not larger than that of $S_2$ (\ie $H_D(S_2,u)$) since $S_1$ is a subset of $S_2$; that is,
	
	\[
	H_D(S_1,u)=\frac{\max_{p\in S_1} u(p)}{\max_{p\in D}u(p)}\leq \frac{\max_{p\in S_2} u(p)}{\max_{p\in D}u(p)}=H_D(S_2,u)
	\]
	
	\noindent \textit{Submodularity:} Given a data point $q\in {D\backslash S_{2}}$, according to the definition of $H_D(S,u)$, we have
	\begin{equation}
	\begin{aligned}
	\triangle'=&H_D(S_1\cup \{q\},u)-H_D(S_1,u) \\
	=&\frac{\max_{p\in S_1\cup \{q\}}u(p)-\max_{p\in S_1}u(p)}{\max_{p\in D}u(p)} \\
	\end{aligned}
	\end{equation}
	\begin{equation}
	\begin{aligned}
	\triangle''=&H_D(S_2\cup \{q\},u)-H_D(S_2,u) \\
	=&\frac{\max_{p\in S_2\cup \{q\}}u(p)-\max_{p\in S_2}u(p)}{\max_{p\in D}u(p)}
	\end{aligned}
	\end{equation}
	
	For convenience, we denote the maximum utility of data points in $S_2 \cup \{q\}$ as \texttt{Maxu}, \ie \texttt{Maxu}$=\max_{p\in S_2\cup \{q\}}u(p)$ and divide the discussion into 3 cases.
	\begin{itemize}
		\item Case 1: If \texttt{Maxu} is achieved by the point in $S_1$; then we have $\triangle'=\triangle'' = 0$.
		\item Case 2: If \texttt{Maxu} is achieved by the point in $S_{2}\backslash S_{1}$,  we have $\triangle' \geq 0$ since the function is monotone and $\triangle'' = 0$. Therefore, $\triangle' \geq \triangle''$.
		\item Case 3: If \texttt{Maxu} is achieved by the point $q$; we have $H_D(S_1\cup \{q\},u)=H_D(S_2\cup \{q\},u)$. Since $S_{1}\subseteq S_{2}$, we have $H_D(S_1,u)\leq H_D(S_2,u)$. Therefore, $\triangle' \geq \triangle''$.
	\end{itemize}

	For all the cases, we have $\triangle' \geq \triangle''$; therefore the function is submodular. Lemma~\ref{le:submo} is correct.	
\end{proof}

Based on Lemma~\ref{le:submo}, $H_D(S,u)$ is submodular for a given utility function $u$. However, according to Lemma~\ref{le:submo_n}, the property does not hold for $H_{D}(S,\mathcal{U})$, where $\mathcal{U}$ represents a class of utility functions.

\begin{lemma}
	\label{le:submo_n}
	The minimum happiness ratio function $H_{D}(S,\mathcal{U})$ is not submodular.
\end{lemma}

\begin{proof}
	We prove Lemma~\ref{le:submo_n} by constructing a counterexample. According to~\cite{Fujishige:2005}, a function $f$ is submodular if it satisfies $f(S_{1})+f(S_{2})\geq f(S_{1}\cup S_{2})+f(S_{1}\cap S_{2})$ for $\forall S_{1},S_{2}\subseteq D$.
	Given a dataset $D=\{(0,1),(1,0)\}$ and a function class $\mathcal{U}$ with two linear functions $\left< 1,0 \right>, \left<0,1 \right>$, let $S_1=\{(0,1)\}$ and $S_2=\{(1,0)\}$. Based on the definition, we have
	$H_D(D,\mathcal{U})=1$ and $H_D(S_1,\mathcal{U})=H_D(S_2,\mathcal{U})=H_D(\emptyset,\mathcal{U})=0$.
	Thus, $1= H_D(S_{1}\cup S_{2},\mathcal{U})+H_D(S_{1}\cap S_{2},\mathcal{U})\nleq H_D(S_{1},\mathcal{U})+H_D(S_{2},\mathcal{U}) = 0$. Therefore, the function is not submodular and Lemma~\ref{le:submo_n} is correct.
\end{proof}

\subsection{Submodularity Ratio and Generalized Curvature}

The submodularity ratio is a quantity between 0 and 1 to characterize how close a set function is to be submodular, while curvature quantifies how close a submodular function is to be modular. Denote $\Delta_{S}(T)=H_{D}(S\cup T, \mathcal{U})-H_{D}(S, \mathcal{U})$. Let $S_0=\emptyset, S_i=\{q_1,q_2,...,q_i\}, i\in[1,k]$ be the successive sets chosen by a greedy algorithm, and $q_i$ denotes the data point selected in step $i$. 

\myparagraphnew{Submodularity ratio}~\cite{Das:2011}. The submodularity ratio of a nonnegative set function $H(\cdot)$ is the largest scalar $\gamma$, \emph{s.t.}
\begin{align}
\label{eq:subratio}
\sum_{p\in T\backslash S}\Delta_{S}(p)\geq \gamma \Delta_{S}(T), \forall T,S\subseteq D.
\end{align}
We can see that the computation of  $\gamma$ is too time consuming. In practice, we only consider the greedy version of Equation \ref{eq:subratio}. The greedy submodularity ratio in \cite{Bian:2017} is the largest scalar $\gamma^G$, \emph{s.t.}
\begin{align*}
\sum_{p\in T\backslash S_i}\Delta_{S_i}(p)\geq \gamma^G\Delta_{S_i}(T), \forall |T|=k, i=0,...,k-1
\end{align*}

\myparagraphnew{Generalized curvature}~\cite{Bian:2017}. The curvature of a nonnegative set function $H(\cdot)$ is the smallest scalar $\alpha$, \emph{s.t.}
\begin{align*}
\Delta_{S\backslash \{p\}\cup T}(p)\geq (1-\alpha)\Delta_{S\backslash \{p\}}(p), \forall T,S\subseteq D,p\in S\backslash T.
\end{align*}
The greedy version of the generalized curvature is the smallest scalar $\alpha^G$, \emph{s.t.}
\begin{align*}
\Delta_{S_{i-1}\cup T}(q_i)\geq (1-\alpha^G)\Delta_{S_{i-1}}(q_i), 
\forall |T|=k, q_i\in S_{k-1}\backslash T.
\end{align*}

%
%
%
%

\section{Accelerated Greedy Algorithms}
\label{sec:sampling}

In this section, we adopt the happiness ratio into the \AlgRDP greedy framework~\cite{Nanongkai:2010} and present some new properties. Then, we fill the gap by conducting strict theoretical analysis for the quality of returned results.
Finally, a sampling-based method is developed to further speed up the search.

\subsection{Greedy Algorithm with Preselection}
\label{sec:alg1}

The $k$-regret minimization problem is NP-hard, and a greedy strategy is often used~\cite{Krause:2014,Bian:2017}. 
After adopting the concept of the happiness ratio, we show that the problem equals finding $k$ data points that can maximize the minimum happiness ratio. 
The objective function $H_D(S,\mathcal{U})$ is monotonic but not submodular. 
The classical greedy algorithm \cite{Nemhauser:1978} can solve the problem, which is critical to efficiently identify the data point with the largest marginal gain.
We call the algorithm by directly using the classical greedy framework as \AlgNWF.
Based on Lemma~\ref{le:1stpoint}, it is time consuming for \AlgNWF to find even the first data point.

\begin{algorithm}[!htb]
	\footnotesize
	\KwIn{A set of $n$ $d$-dimensional points $D=\{p_1, p_2, \cdots, p_n\}$ and $k$.}
	\KwOut{A result set $S_{k}$.}
	\BlankLine
	Let $S_{1}=\{q_1\}$, where $q_1=\arg\max_{p\in D} p[1]$\;
	\For {($i=2$; $i\leq k$; $i++$)}
	{	
		Let $h^*=1$ and $q_i=null$\;
		\For {each $p_{j}\in {D\backslash S_{i-1}}$}
		{	
			calculate the value of $H_{S_{i-1}\cup \{p_j\}}(S_{i-1},\mathcal{U})$ using linear programming\;
			\If{$h^*>H_{S_{i-1}\cup \{p_j\}}(S_{i-1},\mathcal{U})$}
			{
				$h^*=H_{S_{i-1}\cup \{p_j\}}(S_{i-1},\mathcal{U})$\;
				$q_i=p_j$\;
			}	
		}
		\If{$q_i=null$}{return $S_{i-1}$\;}
		\Else{
			$S_{i}=S_{i-1}\cup \{q_i\}$\;
		}
	}			
	return $S_k$\;
	\caption{\AlgPG $(D, k)$}
	\label{Alg:PG}
\end{algorithm}

\begin{lemma}
	\label{le:1stpoint}
	The time complexity of determining the first data point in \em\AlgNWF\em is at least $O(n^2)$.
\end{lemma}
\begin{proof}
	Based on the \AlgNWF framework, the first data point selected is $q_f=\argmax_{p\in D}H_D(\{p\},\mathcal{U})$, where $H_D(\{p\},\mathcal{U})$ $=\inf_{u\in \mathcal{U}}$ $\frac{u(p)}{\max_{p'\in D} u(p')}$. When the utility function class $\mathcal{U}$ contains only one element $u$, \ie $H_D(\{p\},\mathcal{U})=\frac{u(p)}{\max_{p'\in D} u(p')}$, then we only need to traverse all the data points in $D$ to compute the $H_D(\{p\},\mathcal{U})$.
	However, when facing a class of utility functions (\ie the infinite number of functions), we need to employ a linear programming method to compute the first data point. The time complexity of current state-of-the-art linear programming algorithms is at least $O(n)$ (\eg the time complexity of the Simplex method is $O(n^2d)$). Therefore, the time complexity of determining the first data point $q_f$ is at least $O(n^2)$.
\end{proof}

Based on the \AlgNWF framework, it is still costly to determine even the first data point. To reduce the cost, we adopt the same strategy as \AlgRDP does by preselecting some data points and then integrating the \AlgNWF process.
The details are shown in Algorithm~\ref{Alg:PG}.

As shown in~\cite{Nanongkai:2010}, the data points with the maximum value(s) in the coordinate(s) are more representative\footnote{In the following, we preselect only one data point in the algorithm, in which case the approximation ratio guarantee can be  maximized.}. Therefore, we preselect these data points for our \AlgPG algorithm. Then, our algorithm iteratively selects the point that contributes the most to the happiness ratio $H$ and adds it into the result set (Lines 2-12). The happiness ratio is computed by using the linear programming tools in Line 5.
 
\myparagraph{Computing $H_{S\cup\{p\}}(S, \mathcal{U})$ using an LP}  Given a set $S$ and a point $p$, we can compute $H_{S\cup\{p\}}(S, \mathcal{U})$ using the linear program LP~\ref{eq:lp}.

\begin{equation}
\label{eq:lp}
\begin{aligned}
\min & ~~~~y \qquad \quad  \qquad  \qquad  \qquad  & &  \\
\tiny{s.t.} &~~ \sum_{j=1}^{d}p'[j]v[j]\leq y~~~\forall p'\in S &\\
&~~\sum_{j=1}^{d}p[j]v[j]=1 & &\\
&~~v[j]\geq 0~~~~~~~~~~~~~~~~~\forall j\leq d &\\
&~~y\leq 1 &&
\end{aligned}
\end{equation}

\begin{example}
	We select $9$ points from Table \ref{tab:nbaexample} and use them to illustrate the process of the \AlgPG algorithm. We assume that the class of linear utility functions $\mathcal{U}$ contains 4 linear utility functions ($\ie$ $v_0=\left< 0.9,0.05,0.05 \right>$, $v_{1}=\left<0.05,0.9,0.05 \right>$, $v_{2}=\left<0.05,0.05,0.9\right>$ and $v_{3}=\left<0.33,0.33,0.34\right>$). The attribute values of each point are listed in Table \ref{NBA_hr}, and the corresponding utilities of these points are listed in Table \ref{presgreed}. We run the \AlgPG algorithm to select $k=3$ points, and each step is also listed in Table \ref{presgreed}. Let $H_{i,j}$ denote $H_{S_i\cup \{p_j\}}(S_i,\mathcal{U})$. First, we select the point whose first attribute value is maximal; thus, we select Kevin Durant to $S_1$. Next, for iteration 1, we need to calculate the value of $H_{1,j}$ of each point remaining. For Dwight Howard, $H_{1,9}=H_{S_1\cup \{p_9\}}(S_1,\mathcal{U})=\inf\{\frac{0.96}{\max\{0.96,0.62\}},\frac{0.6}{\max\{0.6,0.95\}},\frac{0.61}{\max\{0.61,0.44\}},\frac{0.72}{\max\{0.72,0.67\}}\}$ $=0.63$. Thus, we select the point Dwight Howard to $S_2$ whose $H_{1,j}$ is minimal. For iteration 2, we also use the same approach to select the point Dwyane Wade to $S_3$. After this iteration, $S_3=\{\text{Kevin Durant, Dwight Howard, Dwyane Wade}\}$ and the minimum happiness ratio of $S_3$ is $H_D(S_3,\mathcal{U})=\inf\{\frac{0.96}{0.96},\frac{0.95}{0.95},\frac{0.73}{0.73},\frac{0.72}{0.72}\}=1$.
\end{example}

\begin{table*}[!h]
	\centering
	\footnotesize
	\caption{Example of the \AlgPG algorithm. We show each step for selecting points by the \AlgPG algorithm}
	\label{presgreed}
	\begin{tabular}{|c|l|cccc||c|cc|cc|c|c|}
		\hline
		\multirow{3}{*}{id} & \multirow{3}{*}{player name} & \multicolumn{4}{c||}{\multirow{2}{*}{utility}} & Step 1 & \multicolumn{4}{c|}{Step 2 (line 2-12)}  &result   &\multirow{3}{*}{$\min$ h.r.} \\ \cline{8-11}
		&  & &    & &  &(line 1) & \multicolumn{2}{c|}{iteration 1} & \multicolumn{2}{c|}{iteration 2}  &(line 13) & \\ \cline{3-12}
		&  &$u_0$ &$u_1$ &$u_2$ &$u_3$ & $S_1$ &$H_{1,j}$  & $S_2$ &$H_{2,j}$  & $S_3$ &$S_3$ & \\
		\hline\hline
		1 &  Kevin	Durant  &0.96 &0.60 &0.61 &0.72  &   &  &  &  &   &  &\\
		2 &  LeBron	James  &0.88 &0.54 &0.67 &0.69  &   &0.91  &  &0.91   &  &  &\\
		3 &  Dwyane	Wade  &0.80 &0.39 &0.73 &0.64  &   &0.84  &  &\textbf{0.84}   &  &  &\\
		4 &  Amare	Stoudemire  &0.74 &0.66 &0.32 &0.57  &   &0.91  &  &1.00   &  &  &\\
		5 &  Zach	Randolph  &0.68 &0.85 &0.46 &0.66  & $\{1\}$   &0.71  &$\{1,9\}$  &1.00  &$\{1,9,3\}$   &$\{1,9,3\}$  &1.00 \\
		6 &  Stephen  Jackson  &0.66 &0.40 &0.68 &0.58  &   &0.90  & &0.90  &   &  &\\
		7 &  David	Lee  &0.66 &0.85 &0.48 &0.66  &   &0.71  &  &1.00  &    &  &\\
		8 &  Monta	Ellis  &0.64 &0.29 &0.73 &0.56  &   &0.84  &  &0.84  &    &  &\\
		9 &  Dwight	Howard  &0.62 &0.95 &0.44 &0.67  &  &\textbf{0.63}  &   &  &  &  &\\
		\hline
		
	\end{tabular}
\end{table*}


\begin{heuristic}
	\label{he:point}
	A point $p\in D\backslash S_{i-1}$ that satisfies $H_D(S_{i-1},\mathcal{U})=H_{S_{i-1}\cup \{p\}}(S_{i-1},\mathcal{U})$ is \textit{w.h.p.} the point contributing the most to $H_{D}(S_{i-1} \cup \{p\}, \mathcal{U})$ of the solution set $S_{i-1}$.
\end{heuristic}

We first show that if $H_D(S_{i-1},\mathcal{U})=H_{S_{i-1}\cup \{p\}}(S_{i-1},\mathcal{U})$ satisfies, the both sides achieve the same minimum happiness ratio value under the same utility function $u$.

The left side of the equality is expanded as follows.
\begin{equation}
	\begin{aligned}
		H_D(S_{i-1},\mathcal{U})=\inf\limits_{u\in \mathcal{U}}\frac{\max_{q\in S_{i-1}}u(q)}{\max_{q\in D}u(q)}=\frac{\max_{q\in S_{i-1}}u_1(q)}{\max_{q\in D}u_1(q)}
	\end{aligned}
\end{equation}
where $u_1$ is the utility function making $H_D(S_{i-1},\mathcal{U})$ achieve the minimal happiness ratio value. Similarly, we assume $u_2$ is the utility function when $H_{S_{i-1}\cup \{p\}}(S_{i-1},\mathcal{U})$ achieves the minimal value. 
\begin{equation}
	\begin{aligned}
		H_{S_{i-1}\cup \{p\}}(S_{i-1},\mathcal{U})&=\inf\limits_{u\in \mathcal{U}}\frac{\max_{q\in S_{i-1}}u(q)}{\max_{q\in S_{i-1}\cup \{p\}}u(q)}\\
		&=\frac{\max_{q\in S_{i-1}}u_2(q)}{\max_{q\in S_{i-1}\cup \{p\}}u_2(q)}\\
	\end{aligned}
\end{equation}

We claim that $u_1=u_2$ and the claim can be proved by contradiction. If $u_1\neq u_2$, we know that
\begin{equation}
	\begin{aligned}
&H_D(S_{i-1},\mathcal{U})=\inf\limits_{u\in \mathcal{U}}\frac{\max_{q\in S_{i-1}}u(q)}{\max_{q\in D}u(q)}\\
&< \frac{\max_{q\in S_{i-1}}u_2(q)}{\max_{q\in D}u_2(q)}< \frac{\max_{q\in S_{i-1}}u_2(q)}{\max_{q\in S_{i-1}\cup \{p\}}u_2(q)}\\
&\neq H_{S_{i-1}\cup \{p\}}(S_{i-1},\mathcal{U}) \quad\quad(\text{Contradict})
	\end{aligned}
\end{equation}

For $u_1=u_2=u$, the equality $\max_{q\in D}u(q)=\max_{q\in S_{i-1}\cup \{p\}}u(q)$ holds. For the happiness ratio function $H$, the aim is to find the inferior value of the fraction of two $\max$ functions and $H$ is also affected by the quality of the subset $S_{i-1}$. We can see that when $S_{i-1}$ is small or chosen badly enough, there will be several inferior values of the happiness function $H$ under different utility functions. Otherwise, if $H_D(S_{i-1},\mathcal{U})<1$, we have $\max_{q\in D}u(q)=u(p)$, \ie $p=\argmax_{q\in D}u(q)$. In practice, we greedily add points to $S_{i-1}$ to make the phenomenon that $H$ has multiple inferior values seldom happen thus Heuristic 1 stands.

Heuristic \ref{he:point} with high probability guarantees that if the equality $H_D(S_{i-1},\mathcal{U})=H_{S_{i-1}\cup \{p\}}(S_{i-1},\mathcal{U})$ is satisfied, we only need to select the point which contributes the most to $H_{D}(S_{i-1} \cup \{p\}, \mathcal{U})$.

\myparagraph{Analysis}
For the solution $S_{i-1}$ at step $i-1$ in Algorithm \ref{Alg:PG}, due to the properties of nonnegative and nondecreasing of the minimum happiness ratio function, the minimum value of the happiness ratio is no less than $0$, \ie $H_{D}(S_{i-1}, \mathcal{U})\geq 0$. Then, the calculation of $\Delta_{S_{i-1}}(p)$ can be simplified to calculate $H_{D}(S_{i-1} \cup \{p\}, \mathcal{U})$ because the value of $H_{D}(S_{i-1})$ remains unchanged. 

According to Heuristic~\ref{he:point}, a data point $p\in D\backslash S_{i-1}$ that contributes the most to the solution set $S_{i-1}$, \ie maximizing $H_D(S_{i-1}\cup \{p\}, \mathcal{U})$, is the point satisfying $H_D(S_{i-1},\mathcal{U})=H_{S_{i-1}\cup \{p\}}(S_{i-1},\mathcal{U})$. This can be done by computing $H_{S_{i-1}\cup \{p\}}(S_{i-1},\mathcal{U})$ for each point $p$ and keeping the point with the minimum value, where $H_{S_{i-1}\cup \{p\}}(S_{i-1},\mathcal{U})$ is computed using the linear program and $p\cdot v$ is normalized due to the scale-invariance property of the $k$-regret query. The time complexity of the linear program is  $O(k^{2}d)$~\cite{Bertsimas:1997} because $p'\in S$ instead of $p'\in D$. Otherwise, the time complexity will be $O(n^{2}d)$. Usually, $|S|<<|D|$; thus, the time complexity decreases by a great extent.

\subsection{Approximation Guarantee with Preselection}

In this section, we conduct a theoretical analysis of the approximation ratio of the returned result set in \AlgPG, which fills the gap in the existing research. Our main theoretical result shows that the result set $S_k$ of \AlgPG achieves a near-optimal solution for general monotone nonsubmodular functions with only one preselected data point. The details are shown in Theorem~\ref{thm:preselection}.
\begin{theorem} \label{thm:preselection}
	Let $S_{i}$ and $i\geq 0$ be the incremental successive solution sets and $\alpha$ and $\gamma$ are the curvature and submodularity ratio of the \em\AlgPG\em algorithm with a cardinality constraint $|S|\leq k$, respectively. Then, for all positive integers $k$, the approximation guarantee achieved by \em\AlgPG\em is $\frac{1}{\alpha}[1-(1-\frac{\alpha\gamma}{k})^{k-1}]$.
\end{theorem}
\begin{proof}
	Let $H(S)=H_D(S, \mathcal{U})$, $\forall S\subseteq D$. According to elementary set theory, we have
	\begin{align}
	\label{eq:the11}
	&H(T\cup S_t)=H(T)+\sum_{q_i\in S_t}\Delta_{T\cup S_{i-1}}(q_i)  \nonumber \\
	&=H(T)+\sum_{q_i\in S_t\backslash T}\Delta_{T\cup S_{i-1}}(q_i)+\sum_{q_i\in S_t\cap T}\Delta_{T\cup S_{i-1}}(q_i)  \nonumber \\
	&=H(T)+\sum_{q_i\in S_t\backslash T}\Delta_{T\cup S_{i-1}}(q_i)
	\end{align}
	From the definition of the submodularity ratio, we have
	\begin{align}
	\label{eq:the12}
	H(T\cup S_t)\leq H(S_t)+\frac{1}{\gamma}\sum_{p\in T\backslash S_t}\Delta_{S_t}(p)
	\end{align}
	From the definition of curvature, we have,
	\begin{align}
	\label{eq:the13}
	\sum_{q_i\in S_t\backslash T}\Delta_{T\cup S_{i-1}}(q_i)\geq (1-\alpha)\sum_{q_i \in S_t\backslash T}\Delta_{S_{i-1}}(q_i)
	\end{align}
	Note that we use the shorthand $\Delta_i=\Delta_{S_{i-1}}(q_i)$. By combining Equations \ref{eq:the12}, \ref{eq:the13} to \ref{eq:the11}, we have
	\begin{align}
	\label{eq:the14}
	&H(T)= H(T\cup S_t)-\sum_{q_i\in S_t\backslash T}\Delta_{T\cup S_{i-1}}(q_i)  \nonumber \\
	&\leq \alpha\sum_{q_i\in S_t\backslash T}\Delta_i+H(S_t)-\sum_{q_i\in S_t\backslash T}\Delta_i+\frac{1}{\gamma}\sum_{p\in T\backslash S_t}\Delta_{S_t}(p) \nonumber \\
	&=\alpha\sum_{q_i\in S_t\backslash T}\Delta_i+\sum_{q_i\in S_t\cap T}\Delta_i+\frac{1}{\gamma}\sum_{p\in T\backslash S_t}\Delta_{S_t}(p) \nonumber \\
	&\leq \alpha\sum_{q_i\in S_t\backslash T}\Delta_i+\sum_{q_i\in S_t\cap T}\Delta_i+\frac{1}{\gamma}(k-|T\cap S_t|)\Delta_{t+1}
	\end{align}
	When $t=0$, $\Delta_{S_t}(p)\leq \Delta_{t+1}$ does not hold since the first point selected is the point maximizing the first coordinate instead of the first point chosen by the greedy process. In addition, $\Delta_{S_t}(p)\leq \Delta_{t+1}$ \textit{w.h.p.} holds when $t\geq1$ according to Heuristic \ref{he:point}.
	
	We know that $H(S_k)=\sum_{i=1}^{k}\Delta_i$ (telescoping sum). Hence, the approximation ratio is $\frac{H(S_k)}{H(S^*)}=\sum_i\frac{\Delta_i}{H(S^*)}$, which is denoted as $I(\{o_1,...,o_m\})=\sum_i\frac{\Delta_i}{H(S^*)}$. 
	Since $H(\cdot)$ is nondecreasing, according to Equation \ref{eq:the14}, we have
	\begin{align}	
	&\alpha\sum_{q_i\in S_t\backslash S^*}\frac{\Delta_i}{H(S^*)}+\sum_{q_i\in S_t\cap S^*}\frac{\Delta_i}{H(S^*)}+ \nonumber\\
	&\frac{1}{\gamma}(k-|S^*\cap S_t|)\frac{\Delta_{t+1}}{H(S^*)} \geq 1 \nonumber
	\end{align}
	\label{eq:the15}
	where $ t\geq 1$ and let $x_i=\frac{\Delta_i}{H(S^*)}$, we have
	\begin{align}	
	&\alpha\sum_{q_i\in S_t\backslash S^*}x_i+\sum_{q_i\in S_t\cap S^*}x_i+\frac{1}{\gamma}(k-|S^*\cap S_t|)x_{t+1}  \geq 1
	\label{eq:core}
	\end{align}


Based on~\cite{Bian:2017}, we have the following claim held. If $f(S^*)>0$, then $I(\{o_1,...,o_m\})\geq I(\emptyset)$, and the optimal value of $I(\emptyset)$ is $\sum_{i=1}^{k} y_i$, where $y_i=\frac{\gamma}{k}\beta_i$ and $\beta_i=(1-\frac{\gamma\alpha}{k})^{i-1}, i=1,...,k$. The claim indicates that the worst-case approximation ratio occurs when $S_k\cap S^*=\emptyset$. The constructed linear program associated with $I(\emptyset)$ is
\[
I(\emptyset)=\min\sum_{i=1}^kx_i, \  s.t. \  x_i\geq 0
\]
According to Equation \ref{eq:core}, the structure of the constraint matrix in the LP associated with $I(\emptyset)$ is
\begin{equation}
\label{eq:the16matrix}
	\footnotesize
	\begin{bmatrix}
		\alpha & \frac{k}{\gamma} & & & & & &  \\
		\vdots & \vdots & \ddots & & & & \textbf{0} & \\
		\alpha & \alpha & \cdots & \frac{k}{\gamma}  & & & & \\
		\alpha & \alpha & \cdots & \alpha & \frac{k}{\gamma} & & &  \\
		\alpha & \alpha & \cdots & \alpha & \alpha & \frac{k}{\gamma} & &  \\
		\vdots & \vdots & & \vdots & \vdots & \vdots & \ddots & \\
		\alpha & \alpha & \cdots & \alpha & \alpha & \alpha & \cdots & \frac{k}{\gamma} \\
	\end{bmatrix}_{(k-1,k)}
	\cdot
	\begin{bmatrix}
		x_1 \\
		x_2 \\
		\vdots \\
		x_a \\
		x_b \\
		x_c \\
		\vdots \\
		x_k \\
	\end{bmatrix}
	\geq
	\begin{bmatrix}
		1 \\
		\vdots \\
		1 \\
		1 \\
		1 \\
		\vdots \\
		1 \\
	\end{bmatrix}
\end{equation}

Let $y_1=\beta_1$. Then, $y_i=\frac{\gamma}{k}\beta_i$,  $\beta_i=(1-\alpha\beta_{1})(1-\frac{\alpha\gamma}{k})^{i-2}, i=2,...,k$. We can observe that the vector $\textbf{y}\in \mathbb{R}_+^k$ satisfies all the constraints, and each row in Equation \ref{eq:the16matrix} is tight; hence, \textbf{y} is the optimal solution. Therefore,
\[
I(\emptyset)=\sum_{i=1}^ky_i=\beta_{1}+\frac{1}{\alpha}(1-\alpha\beta_{1})[1-(1-\frac{\alpha\gamma}{k})^{k-1}].
\]

We know that $\beta_{1}\geq0$ and function $I(\emptyset)$ is monotonically increasing function of $\beta_{1}$. Hence, $I(\emptyset)$ is no smaller than $\frac{1}{\alpha}[1-(1-\frac{\alpha\gamma}{k})^{k-1}]$, and  $\frac{1}{\alpha}[1-(1-\frac{\alpha\gamma}{k-1})^{k-1}] \ge  \frac{1}{\alpha}[1-(1-\frac{\alpha\gamma}{k})^{k-1}]$.

In summary, $I(\{o_1,...,o_m\})\geq I(\emptyset)\geq \frac{1}{\alpha}[1-(1-\frac{\alpha\gamma}{k})^{k-1}]$, \ie $H_D(S_k, \mathcal{U})\geq \frac{1}{\alpha}[1-(1-\frac{\alpha\gamma}{k})^{k-1}]H_D(S^*, \mathcal{U})$. The theorem is correct.
\end{proof}

From Theorem \ref{thm:preselection}, we can see that when $\alpha$ and $\gamma$ approach 1, the approximation ratio is closely equal to $1-1/e$ with not too small $k$, \eg $k>5$. Even though $\alpha$ and $\gamma$ take values from the range $[0,1]$, Theorem \ref{thm:preselection} provides an approximation ratio of Algorithm \ref{Alg:PG}, which guarantees that the \AlgNWF framework cannot be arbitrarily bad.

\subsection{Accelerated Processing via Sampling}
\label{sec:alg2}

Even though the \AlgPG algorithm can efficiently compute the result with theoretical guarantees, it is still not scalable when $n$ is large. To scale for large datasets, a common method uses a sampling-based technique to reduce the number of data points evaluated~\cite{Mirzasoleiman:2015} and returns a result with tight theoretical guarantees.
Therefore, we propose a sampling-based method, \AlgSPG, which extends the proposed \AlgPG algorithm.
The algorithm details are shown in Algorithm~\ref{Alg:SPG}, which replaces lines 3-4 in Algorithm~\ref{Alg:PG} with lines 1-3. 
The difference between \AlgSPG and \AlgPG is that \AlgPG finds a point from $D\backslash S_{i-1}$ directly, while \AlgSPG samples a subset $R$ randomly and
then finds the point in $R$ which contributes the most to the value of the happiness ratio. To bound the quality of returned results for a sampling-based approach, a critical issue is to carefully choose the sample size. To do this, in each iteration of the algorithm, we sample a set $R$ of size $s=\frac{n}{k}\log(\frac{\lambda}{\lambda-1+\epsilon})$ uniformly at random, where $\epsilon > 0$ and $\lambda\geq 1$, which is the sample factor. 

The proposed algorithm is simple, but it can provide high-quality results, which is verified by the following concrete example, theoretical analysis and our experimental results. 

\begin{algorithm}[!htb]
	\vspace{0.6mm} \nonl \textsc{\underline{// Replace Line 3-4 in Alg.~\ref{Alg:PG} with following code}} \\ \vspace{0.6mm}
	Let $h^*=1$ and $q_i=null$\;
	Obtain a random subset $R$ by sampling $s$ random points from $D\backslash S_{i-1}$\;
	\textbf{for} $each$ $p_{j}\in {R}$ \textbf{do}
	\caption{\AlgSPG $(D, k)$}
	\label{Alg:SPG}
\end{algorithm}

\begin{table}
	\centering
	\setlength{\tabcolsep}{0.5mm}
	\scriptsize
	\caption{Example of the \AlgSPG algorithm. We run \AlgSPG algorithm 4 times and show each step of selecting points in each time}
	\label{stocpresgreed}
	\begin{tabular}{|c|c|ccc|ccc|c|c|}
		\hline
		\multirow{3}{*}{Alg.} & \multirow{2}{*}{Step 1} & \multicolumn{6}{c|}{Step 2}  &\multirow{2}{*}{result}   &\multirow{3}{*}{$\min$ h.r.} \\ \cline{3-8}
		& & \multicolumn{3}{c|}{iteration 1} & \multicolumn{3}{c|}{iteration 2}  & & \\ \cline{2-9}
		& $S_1$ &samples &$H_{1,j}$  & $S_2$ &samples &$H_{2,j}$  & $S_3$ &$S_3$ & \\
		\hline\hline
		\multirow{6}{*}{SPG1}&  \multirow{6}{*}{$\{1\}$} & 2 &0.91  &\multirow{6}{*}{$\{1,9\}$}  &2 &0.91 &\multirow{6}{*}{$\{1,9,3\}$} &\multirow{6}{*}{$\{1,9,3\}$}  &\multirow{6}{*}{$1.00$}   \\
		&   & 3 &0.84  &  &3 &\textbf{0.84} & &  &   \\
		&   & 6 &0.90 &  &4 &1.00 & &  &   \\
		&   & 7 &0.71  &  &5 &1.00 & &  &   \\
		&   & 8 &0.84  &  &6 &0.90 & &  &   \\
		&   & 9 &\textbf{0.63}  &  &8 &0.84 & &  &   \\ \hline
		\multirow{5}{*}{SPG2}&  \multirow{5}{*}{$\{1\}$} & 2 & 0.91 & \multirow{5}{*}{$\{1,9\}$} &2 &0.91 &\multirow{5}{*}{$\{1,9,8\}$} &\multirow{5}{*}{$\{1,9,8\}$}  &\multirow{5}{*}{$1.00$}   \\
		&   & 5 & 0.71 &  &4 &1.00 & &  &   \\
		&   & 6 & 0.90 &  &5 &1.00 & &  &   \\
		&   & 7 & 0.71 &  &6 &0.90 & &  &   \\
		&   & 9 & \textbf{0.63} &  &8 &\textbf{0.84} & &  &   \\ \hline
		\multirow{4}{*}{SPG3}&  \multirow{4}{*}{$\{1\}$} & 2 & 0.91 &\multirow{4}{*}{$\{1,7\}$}  &3 &\textbf{0.84} &\multirow{4}{*}{$\{1,7,3\}$} &\multirow{4}{*}{$\{1,7,3\}$}  &\multirow{4}{*}{$0.89$}   \\
		&   & 4 & 0.91 &  &4 &1.00 & &  &   \\
		&   & 7 & \textbf{0.71} &  &5 &1.00 & &  &   \\
		&   & 8 & 0.84 &  &9 &0.89 & &  &   \\ \hline
		\multirow{3}{*}{SPG4}&  \multirow{3}{*}{$\{1\}$} & 2 & 0.91 &\multirow{3}{*}{$\{1,7\}$}  &2 &\textbf{0.91} &\multirow{3}{*}{$\{1,7,2\}$} &\multirow{3}{*}{$\{1,7,2\}$}  &\multirow{3}{*}{0.89}   \\
		&   & 6 & 0.90 &  &4 &1.00 & &  &   \\
		&   & 7 & \textbf{0.71} &  &5 &1.00 & &  &   \\
		\hline
		
	\end{tabular}
\end{table}

\begin{example}
	\label{SPGExample}
	We select the same $9$ points as in Table \ref{presgreed} and use them to illustrate the process of the \AlgSPG algorithm. We also assume that the class of linear utility functions $\mathcal{U}$ contains 4 linear utility functions ($\ie$ $v_0=\left< 0.9,0.05,0.05 \right>$, $v_{1}=\left<0.05,0.9,0.05 \right>$, $v_{2}=\left<0.05,0.05,0.9\right>$ and $v_{3}=\left<0.33,0.33,0.34\right>$). Therefore, the attribute values of each point are listed in Table \ref{NBA_hr}, and the corresponding utilities of the 4 utility functions are the same as in Table \ref{presgreed}. We run our \AlgSPG algorithm 4 times, and  the sample sizes are 6, 5, 4 and 3. We aim to select $k=3$ points to the result set, and the results are shown in Table \ref{stocpresgreed} corresponding to algorithms SPG1$\sim$4. Let $H_{i,j}$ denote $H_{S_i\cup \{p_j\}}(S_i,\mathcal{U})$ for simplicity. For SPG1, first, we select Kevin Durant to $S_1$. Next, for iteration 1, we sample 6 points from the remaining points and calculate the values of $H_{1,j}$ of these 6 points. Then we select the point Dwight Howard to $S_2$ whose $H_{1,j}$ is minimal. For iteration 2, we also sample 6 points from the remaining points, and use the same approach to select the point Dwyane Wade to $S_3$. After this iteration, we select 3 points to result set $S_3$; then, the minimum happiness ratio of $S_3$ is 1. The processes of SPG2$\sim$4 are the same as SPG1.
	
	We can see that with 2 times, the minimum happiness ratios are equal to 1, which is the same as the result of \AlgPG. The average minimum happiness ratio of the 4 algorithms is 0.945, which is very close to 1. If we choose the maximum value among the results, it has a great chance to select a result as good as that of \AlgPG. In addition, the result of random sampling is also influenced by $k$. When sampling a subset $R$ with size $s$ from the dataset $D$, the probability of a data point $p$ being selected is $Pr(p)=1-(C_{n-1}^s/C_n^s)^k$ for sampling without replacement, which approaches 1 when $k$ increases.
\end{example}

A theorem shows that the sample size used in Algorithm~\ref{Alg:SPG} is sufficient and the returned result is theoretically bounded. 
Before we introduce the details of the theorem, we first show that the expected marginal gain of the algorithm in each iteration is bounded in Lemma~\ref{le:SPGstep}.

\begin{lemma}
	\label{le:SPGstep}
	Given the current solution $S$ and $\epsilon > 0$, $\lambda\ge1$, the expected marginal gain of \em\AlgSPG\em in each iteration (except the first step) is at least $\frac{1-\epsilon}{\lambda k}\sum_{p\in S^*\backslash S}\Delta_S(p)$.
\end{lemma}

\begin{proof}	
	To prove the lemma, we need to estimate the probability that $R\cap (S^*\backslash S)\neq \emptyset$. The set $R$ consists of $s$ random samples from $D\backslash S$, hence
	\begin{align*}
	Pr[R\cap (S^*\backslash S)=\emptyset]&=\left( 1-\frac{|S^*\backslash S|}{|D\backslash S|} \right)^s  \\
	&\leq e^{-s\frac{|S^*\backslash S|}{|D\backslash S|}}\leq e^{-\frac{s}{n}|S^*\backslash S|}
	\end{align*}
	Therefore, by using the concavity of $1-e^{-\frac{s}{n}x}$ as a function of $x$ and the fact that $x=|S^*\backslash S|\in [0, k]$, we obtain
	\begin{align*}
	Pr[R\cap (S^*\backslash S)\neq \emptyset]\geq 1-e^{-\frac{s}{n}|S^*\backslash S|}\geq (1-e^{-\frac{sk}{n}})\frac{|S^*\backslash S|}{k}.
	\end{align*}
	Recall that we set $s=\frac{n}{k}\log(\frac{\lambda}{\lambda-1+\epsilon})$, which gives
	\begin{align}
	\label{eq:the21}
	Pr[R\cap (S^*\backslash S)\neq \emptyset]\geq (1-\epsilon)\frac{|S^*\backslash S|}{\lambda k}.
	\end{align}
	In algorithm \AlgSPG, we select an element $p\in R$ that contributes the most to the solution set $S$. It is clear that the marginal value of $p$ is no less than the maximum contribution of an element randomly chosen from $R\cap (S^*\backslash S)$ (if not empty). Overall, $R$ is equally likely to contain each element of $S^*\backslash S$, so a randomly selected element from $R\cap (S^*\backslash S)$ is actually a randomly selected element from $S^*\backslash S$. Thus, we obtain
	\begin{align*}
	\mathbf{E}[\Delta_S(p)]\geq Pr[R\cap (S^*\backslash S)\neq \emptyset]\times \frac{1}{|S^*\backslash S|}\sum_{p\in S^*\backslash S}\Delta_S(p).
	\end{align*}
	Hence, by combining Equation \ref{eq:the21}, the lemma is correct, \ie $\mathbf{E}[\Delta_S(p)]\geq \frac{1-\epsilon}{\lambda k}\sum_{p\in S^*\backslash S}\Delta_S(p)$.
\end{proof}

\begin{theorem}\label{thm:sampling}
	Given $\epsilon >0$, $\lambda\ge1$ and $s=\frac{n}{k}\log(\frac{\lambda}{\lambda-1+\epsilon})$, then \em\AlgSPG\em can achieve an $(1-e^{-\frac{(1-\epsilon)(k-1)\gamma}{\lambda k}})$ approximation ratio guarantee in expectation.
\end{theorem}

\begin{proof}
	Let $S_i=\{q_1,...,q_i\}$ denote the solution returned by \AlgSPG after $i$ steps, $i\geq 1$. From Lemma~\ref{le:SPGstep},
	\begin{align}
	\label{eq:the22}
	\mathbf{E}[\Delta_{S_i}(q_{i+1})|S_i]\geq \frac{1-\epsilon}{\lambda k}\sum_{p\in S^*\backslash S_i}\Delta_{S_i}(p).
	\end{align}
	Based on Equation \ref{eq:subratio}, we have
	\begin{align*}
		\small
	\sum_{p\in S^*\backslash S_i}\Delta_{S_i}(p)\geq \gamma\Delta_{S_i}(S^*) =\gamma(H_D(S^*\cup S_i, \mathcal{U})-H_D(S_i, \mathcal{U}))
	\end{align*}
	Based on the monotonic property, we have
	\begin{align*}
	H_D(S^*\cup S_i, \mathcal{U})-H_D(S_i, \mathcal{U})\geq H_D(S^*, \mathcal{U})-H_D(S_i, \mathcal{U})
	\end{align*}
	Based on Equation \ref{eq:the22}, we can obtain
	\begin{align*}
	\mathbf{E}[\Delta_{S_i}(q_{i+1})|S_i]&=\mathbf{E}[H_D(S_{i+1}, \mathcal{U})-H_D(S_i, \mathcal{U})|S_i] \\
	&\geq \frac{1-\epsilon}{\lambda k}\gamma(H_D(S^*, \mathcal{U})-H_D(S_i, \mathcal{U}))
	\end{align*}
	By taking expectation over $S_i$, we have
	\begin{align*}
	&\mathbf{E}[H_D(S_{i+1}, \mathcal{U})-H_D(S_i, \mathcal{U})]\\
	&\geq \frac{1-\epsilon}{\lambda k}\gamma \mathbf{E}[H_D(S^*, \mathcal{U})-H_D(S_i, \mathcal{U})]
	\end{align*}
	
	Note that in our \AlgSPG algorithm, the first point does not contribute to the theoretical guarantee of the lower bound. Therefore, by induction, it implies that
	\begin{align*}
	\mathbf{E}[H_D(S_k, \mathcal{U})] &\geq (1-(1-\frac{1-\epsilon}{\lambda k}\gamma)^{k-1})\mathbf{E}[H_D(S^*, \mathcal{U})] \\
	&\geq (1-e^{-\frac{(1-\epsilon)(k-1)\gamma}{\lambda k}})\mathbf{E}[H_D(S^*, \mathcal{U})]
	\end{align*}
	
	Thus, the theorem is proved.
\end{proof}

In Theorem 2, $\lambda$ is a tradeoff parameter between the sample size and the corresponding approximation ratio.
A larger $\lambda$ denotes a smaller sample size and vice versa. When $\lambda =1$, we tend toward the bound $(1-e^{-\frac{(1-\epsilon)(k-1)\gamma}{k}})$, which extends the bound in \cite{Mirzasoleiman:2015} for monotonic functions from submodular to nonsubmodular and the bound of greedy algorithms in \cite{Khanna:2017} for nonsubmodular (weakly submodular) functions from deterministic to stochastic. The bound also covers the well-known bound $(1-1/e)$ \cite{Nemhauser:1978} for submodular functions. 

Though Theorems \ref{thm:preselection} and \ref{thm:sampling} provide parameterized theoretical guarantees of Algorithms \ref{Alg:PG} and \ref{Alg:SPG}, respectively, in most cases, they approach $1-1/e$ from a theoretical aspect. In practice, they usually approach 1, which can also be verified in our experiments. The proposed theorems guarantee that the \AlgNWF frame with/without sampling cannot be arbitrarily poor.

\section{Empirical Evaluation}
\label{sec:exp}

In this section, we present the results of a comprehensive performance study on both synthetic and real-world datasets to evaluate the efficiency and effectiveness of the techniques proposed in this paper.

\subsection{Experimental Setup}

\myparagraph{Datasets} We run our experiments on 1 synthetic and 4 real-world datasets, which are widely used in previous studies, \eg~\cite{Nanongkai:2010,Peng:2014,Nanongkai:2012,Faulkner:2015,Qi:2018,Xie:2018}. Moreover, similar to studies in the literature \cite{Nanongkai:2010,Peng:2014,Nanongkai:2012,Faulkner:2015,Xie:2018} on $k$-regret queries, we computed the skyline first and then identified $k$ points from it.
The information about the datasets is summarized in Table~\ref{realdatasets}.

\begin{itemize}
	\item \textbf{Anti-correlated} (synthetic dataset). The dataset is created by using the dataset generator in \cite{Borzsony:2001}. It is a 6-dimensional anti-correlated dataset with 10,000 data points and 5,531 skyline points.
	\item The \textbf{NBA}\footnote{\url{https://www.rotowire.com/basketball/}}. NBA dataset is extracted from NBA players' game statistics from 1946 to 2009 with 21,961 records and 164 skyline records, each of which has 8 dimensions.
	\item \textbf{Household}\footnote{\url{http://www.ipums.org}}. It is a 6-dimensional dataset with 127,932 records and 49 skyline records, each of which represents the percentage of an American family's annual income on  different types of expenditures.
	\item \textbf{Movie}\footnote{\url{https://movielens.umn.edu}}. It is a 21-dimensional dataset with 100,000 ratings (1-5) from 943 users on 1,682 movies along with 470 skyline ratings.
	\item \textbf{Weather}\footnote{\url{http://cru.uea.ac.uk/cru/data/hrg/tmc/}}. It is 15-dimensional weather data of 566,268 points with 63,398 skyline points, each of which consists of average monthly precipitation totals and elevation at over half a million sensor locations.
\end{itemize}

\begin{table}[t]
	\centering
	\caption{Dataset Statistics}
	\label{realdatasets}
	\footnotesize
	\renewcommand{\arraystretch}{1.2}
	\begin{tabular}{|c | c | c |}\hline
		\textbf{Dataset}      &\textbf{Dimensionality}    &\textbf{Size}    \\ \hline\hline
		\textbf{Anti-correlated}      &$6$   &$10,000$   \\ \hline
		\textbf{NBA}  &$8$ &$21,961$ \\ \hline
		\textbf{Household}  &$6$  &$127,391$  \\ \hline
		\textbf{Movie}  &$21$  &$1,682$  \\ \hline
		\textbf{Weather}  &$15$  &$566,268$  \\ \hline
	\end{tabular}
\end{table}

\myparagraph{Algorithms} We implement and evaluate the following algorithms in the experiments.
\begin{itemize}
	\item \textbf{\AlgNWF}. Baseline algorithm using the classical greedy framework~\cite{Nemhauser:1978}.
	\item \textbf{\AlgHDR}. The set-cover based algorithm~\cite{Asudeh:2017}.
	\item \textbf{\AlgHDG}. The greedy-based heuristic algorithm \cite{Asudeh:2017}.
	\item \textbf{\AlgKER}. The coreset-based algorithm~\cite{Cao:2017}.
	\item \textbf{\AlgHS}. The hitting set algorithm~\cite{Agarwal:2017}.
	\item \textbf{\AlgSPH}. The state-of-the-art algorithm~\cite{Xie:2018}.
	\item \textbf{\AlgPG}. The greedy algorithm adopts the happiness ratio, which is proposed in Section~\ref{sec:alg1}. 
	\item \textbf{\AlgSPG}. The sampling-based algorithm proposed in Section~\ref{sec:alg2}.
\end{itemize}

Note that there exist improved versions of the \textsc{Cube} algorithm \cite{Nanongkai:2010} and GeoGreedy algorithm \cite{Peng:2014}, namely \AlgKER and \AlgSPH, respectively; we only compare with the improved algorithms.

\myparagraph{Workload and implementation}
To evaluate the performance of the proposed techniques, we report the maximum regret ratio and CPU time by varying $k$, $d$ and $n$. Note that the maximum regret ratios shown in the experiments are calculated by subtracting the corresponding minimum happiness ratios from 1.
We also report the submodularity ratio and curvature for different algorithms.
We only calculate the greedy version of $\gamma^G$ and $\alpha^G$ because it is too time consuming to calculate the full version by exhaustive search.
To complete the process within an acceptable time, $k\in[2,6]$ are selected for the submodularity ratio and curvature experiments.

For \AlgSPG, we only use $\epsilon=0.01,\lambda=1.01$ and $\epsilon=0.1,\lambda=1.1$ to demonstrate its performance named $SPG1$ and $SPG4$, respectively (similar setting as in Example \ref{SPGExample}). From Lemma \ref{le:SPGstep} and Theorem \ref{thm:sampling}, the adopted $\lambda$ values are only slightly larger than 1, which can greatly reduce the dataset size, while there is little change in the approximation ratio. This also shows the advantage of our \AlgSPG algorithm.

We run each setting 20 times and report the average value.

All the algorithms are implemented in C++ with GCC 4.8.5. Experiments are conducted on a workstation with a 3.3GHz CPU using an Ubuntu 16.04 LTS.
Linear programming is implemented with the GNU Linear Programming Kit~\footnote{https://www.gnu.org/software/glpk/}.

\begin{figure*}[!htb]
	\centering
	\includegraphics[width=1\textwidth]{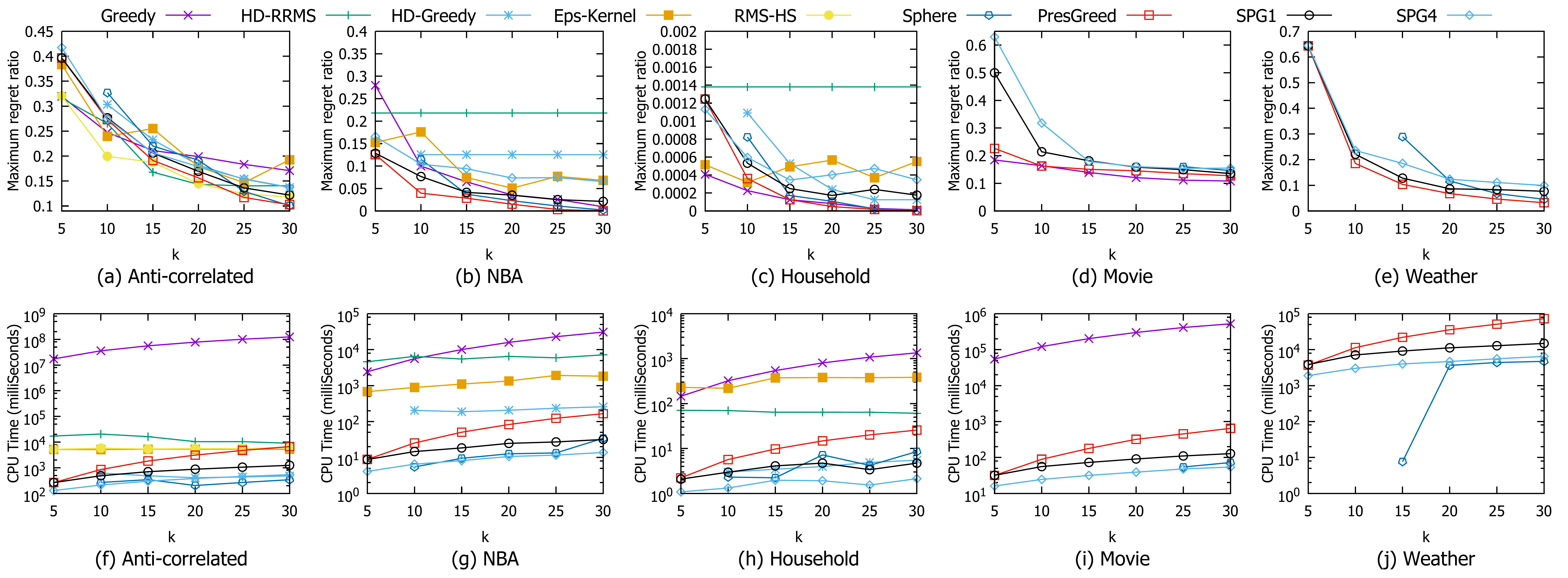}
	\caption{Performance comparisons. (\emph{a})-(\emph{e}) show the maximum regret ratios of all the algorithms for different values of $k$, and (\emph{f})-(\emph{j}) show the CPU times for different values of $k$.
	}
	\label{fig:MHTk}
\end{figure*}

\begin{figure*}[!htb]
	\setlength{\abovecaptionskip}{0cm}
	\setlength{\belowcaptionskip}{-0.5cm}
	\centering
	\includegraphics[width=0.98\textwidth]{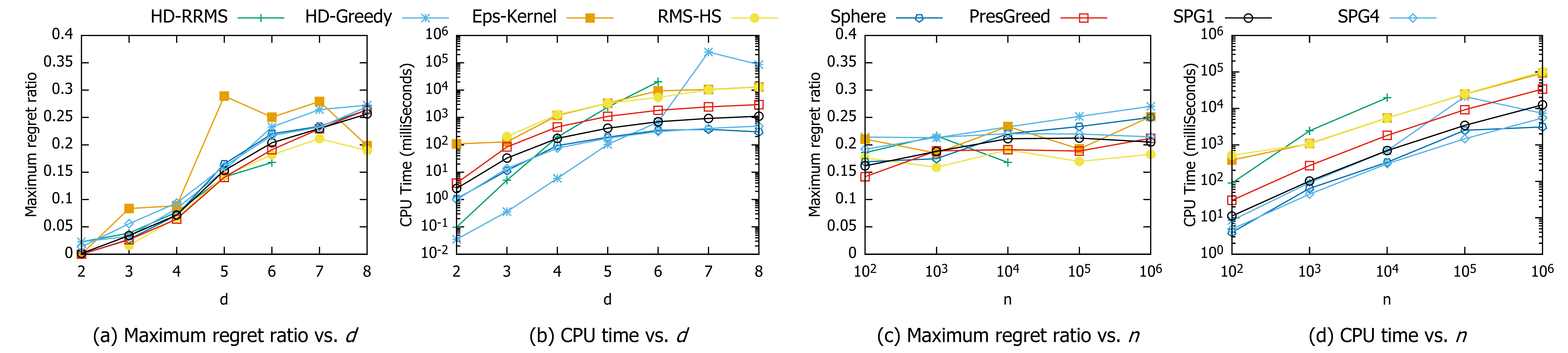}
	\caption{Performance comparisons. (\emph{a}), (\emph{c}) and (\emph{b}), (\emph{d}) show the maximum regret ratios and CPU times of all the algorithms for different values of $d$ and $n$ on the anti-correlated dataset respectively.
	}
	\label{fig:MHTdn}
\end{figure*}

\begin{figure*}[!htb]
	\setlength{\abovecaptionskip}{0cm}
	\setlength{\belowcaptionskip}{-0.5cm}
	\centering
	\subfigure[No Sampling]{\includegraphics[width=0.19\textwidth]{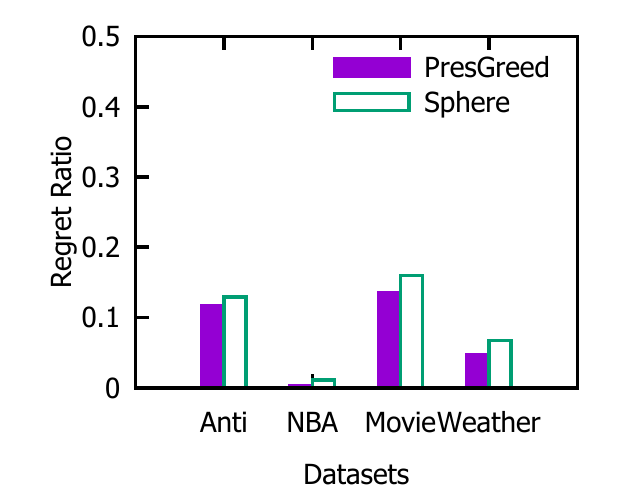}}
	\subfigure[Anti-correlated]{\includegraphics[width=0.19\textwidth]{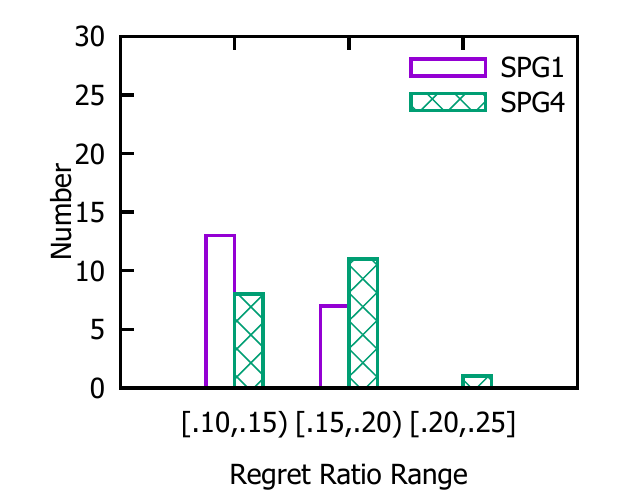}}
	\subfigure[NBA]{\includegraphics[width=0.19\textwidth]{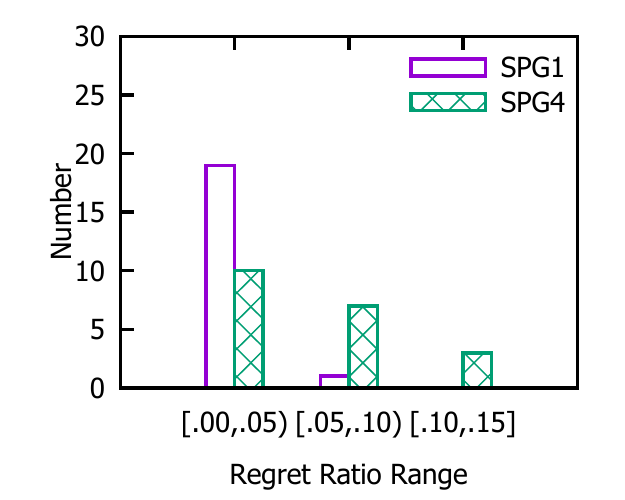}}
	\subfigure[Movie]{\includegraphics[width=0.19\textwidth]{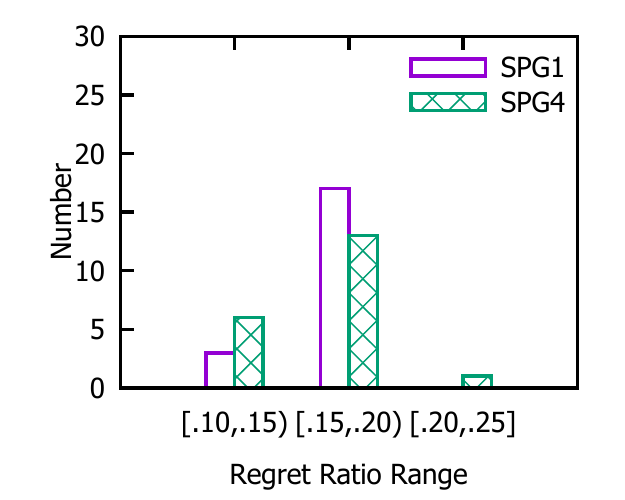}}
	\subfigure[Weather]{\includegraphics[width=0.19\textwidth]{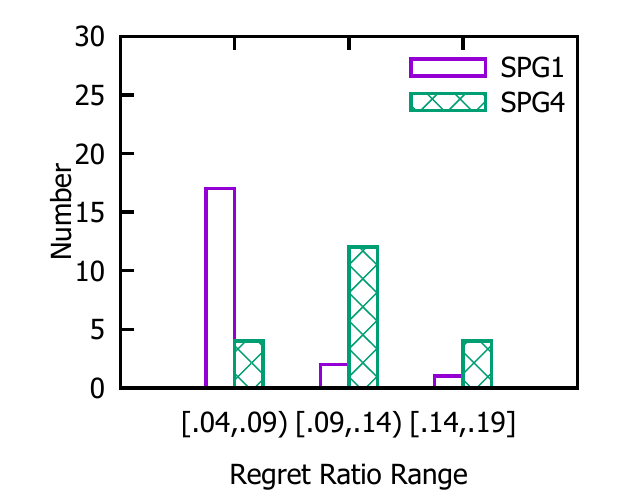}}	
	\caption{Distributions of regret ratios. (\emph{a}) shows the regret ratios on 4 datasets of \AlgPG without sampling, and (\emph{b})-(\emph{e}) show the regret ratio distributions of 20 repeated experiments of the \AlgSPG algorithm at $k=25$ on these datasets.
	}
	\label{fig:MHTdis}
\end{figure*}

\begin{figure*}[!htb]
	\setlength{\abovecaptionskip}{0cm}
	\setlength{\belowcaptionskip}{-0.5cm}
	\centering	
	\subfigure[Anti-correlated]{\includegraphics[width=0.19\textwidth]{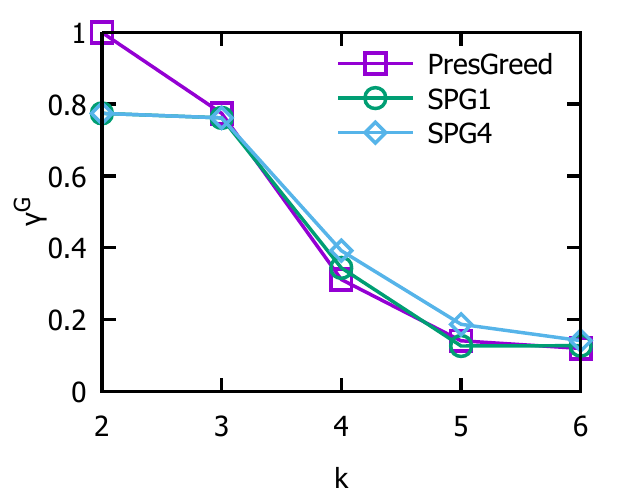}}
	\subfigure[NBA]{\includegraphics[width=0.19\textwidth]{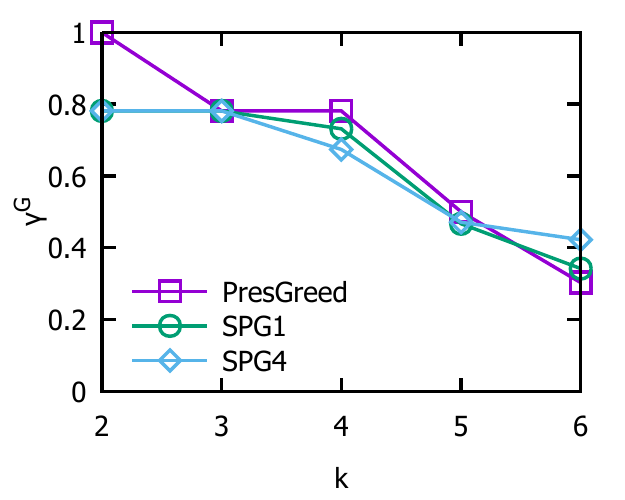}}
	\subfigure[Household]{\includegraphics[width=0.19\textwidth]{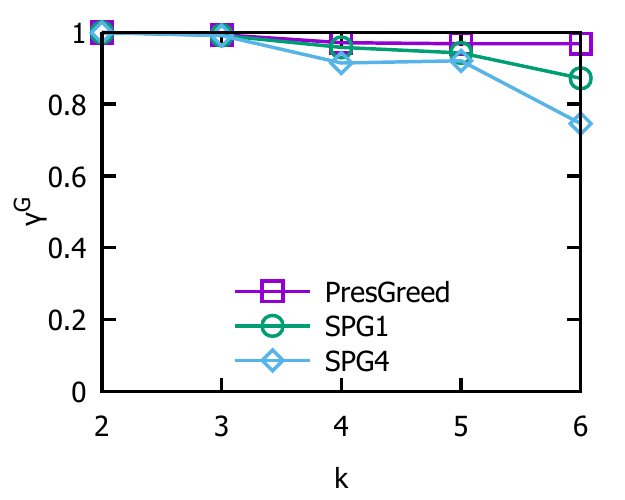}}
	\subfigure[Movie]{\includegraphics[width=0.19\textwidth]{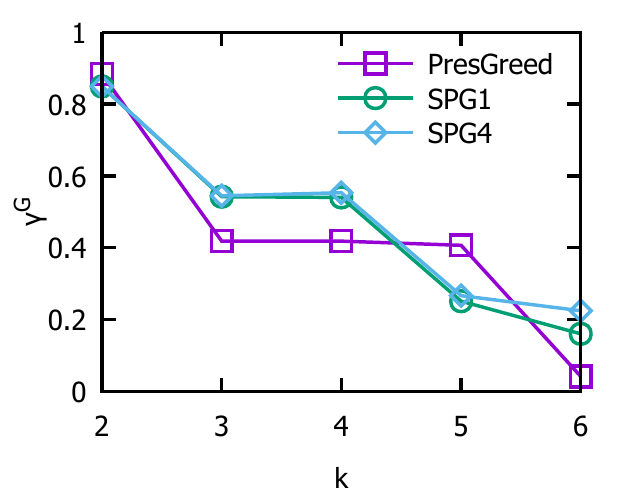}}
	\subfigure[Weather]{\includegraphics[width=0.19\textwidth]{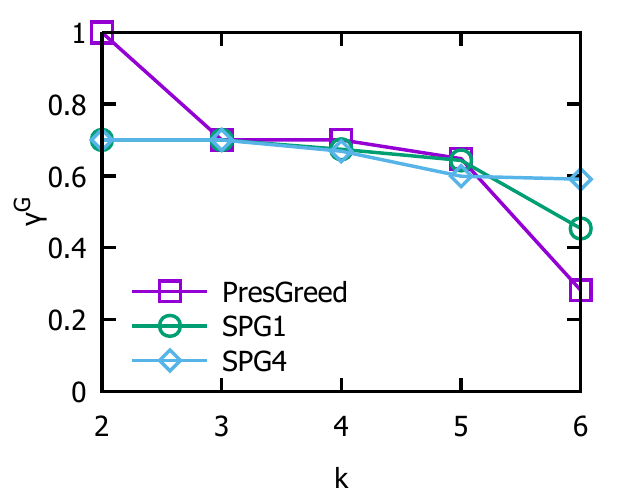}}
	\\
	\subfigure[Anti-correlated]{\includegraphics[width=0.195\textwidth]{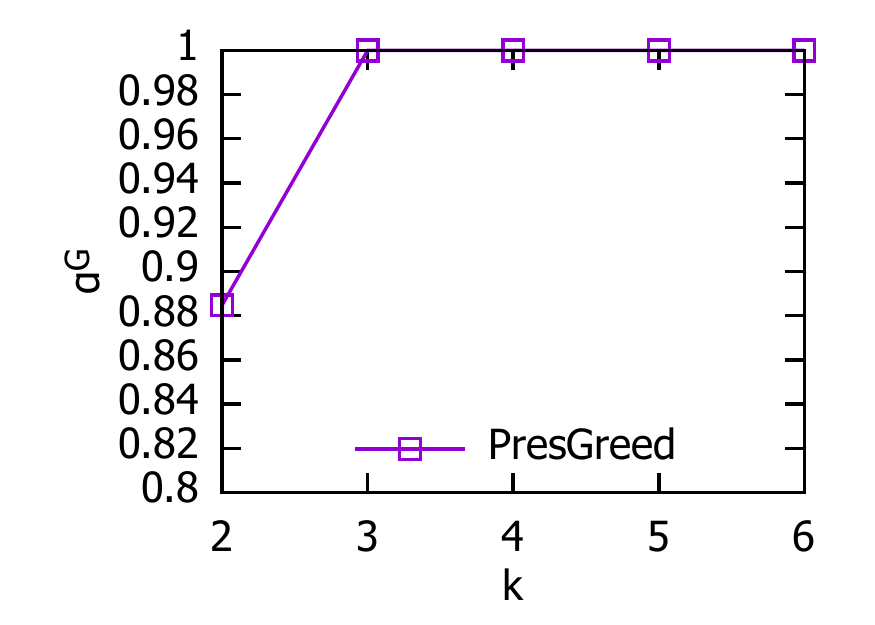}}
	\subfigure[NBA]{\includegraphics[width=0.195\textwidth]{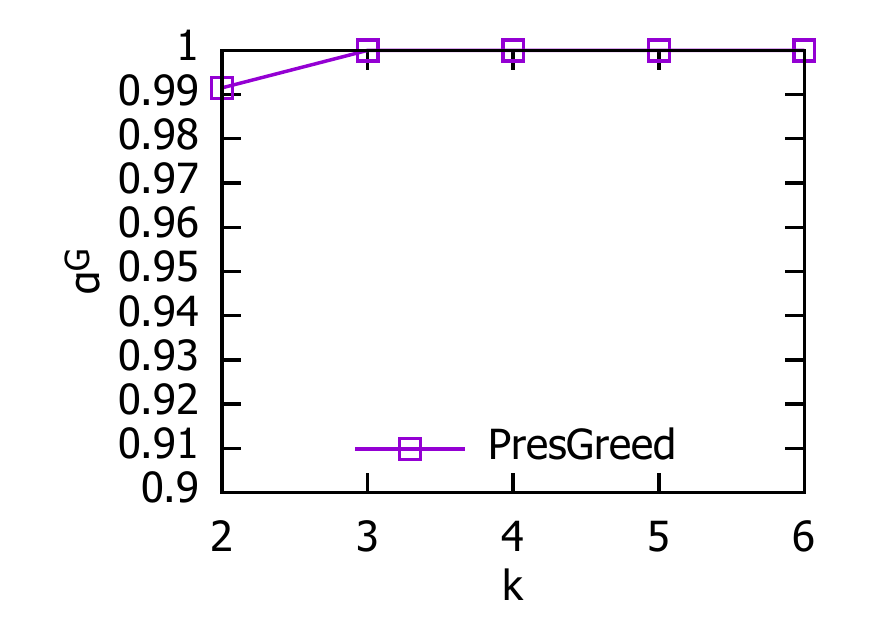}}
	\subfigure[Household]{\includegraphics[width=0.195\textwidth]{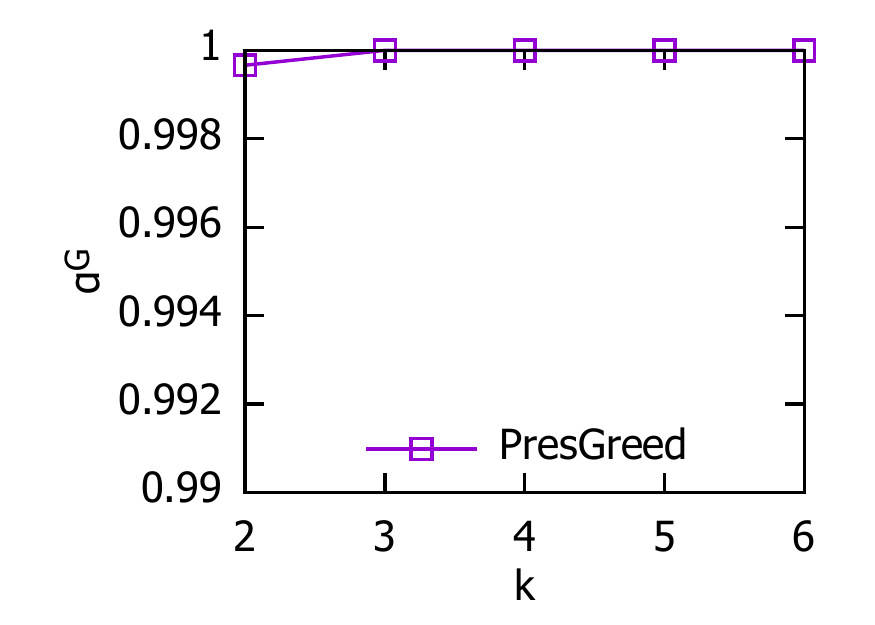}}
	\subfigure[Movie]{\includegraphics[width=0.195\textwidth]{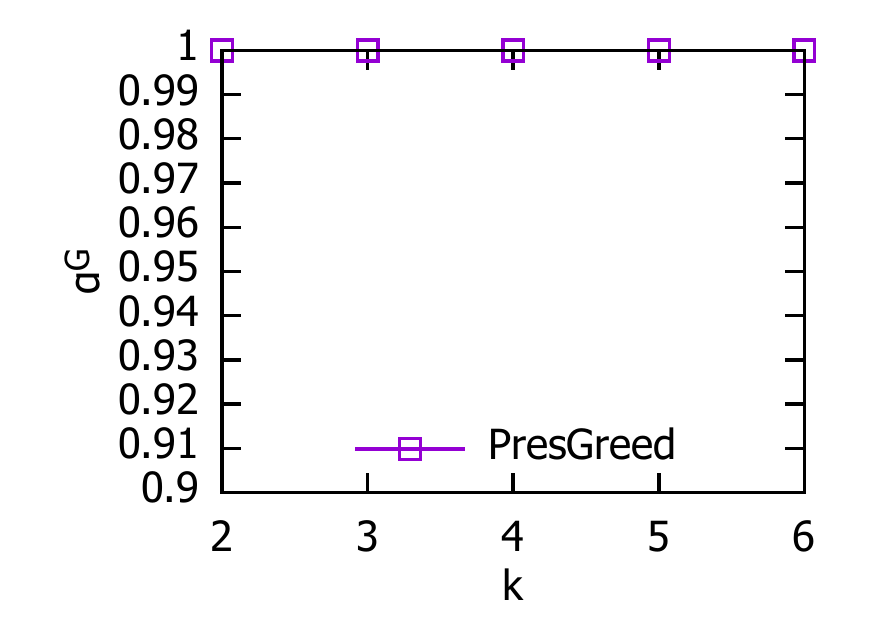}}
	\subfigure[Weather]{\includegraphics[width=0.195\textwidth]{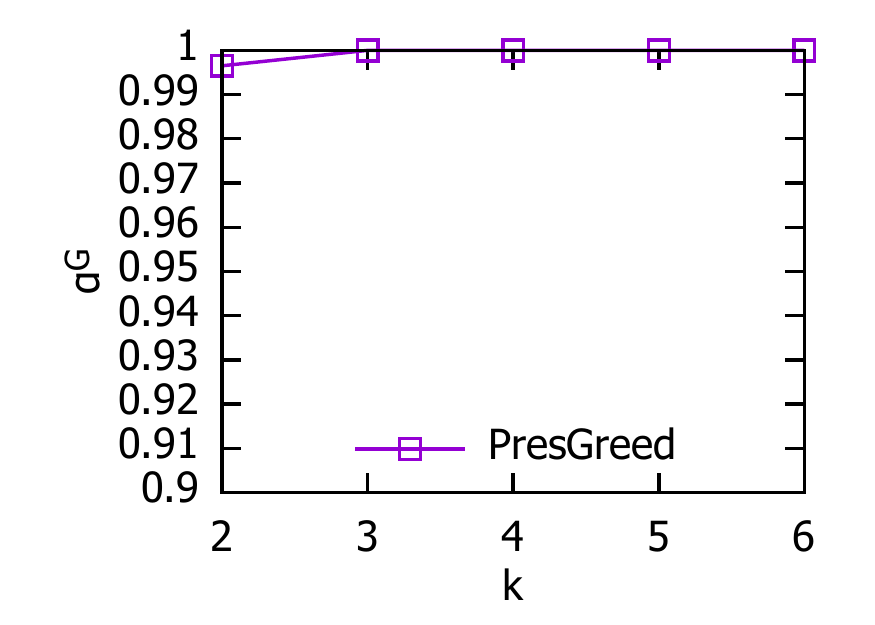}}	
	\caption{Performance comparisons. (\emph{a})-(\emph{e}), (\emph{f})-(\emph{j}) show the submodularity ratios and curvatures respectively for different values of $k$ on \textit{Anti-correlated}, \textit{NBA}, \textit{Household}, \textit{Movie} and \textit{Weather}.}
	\label{fig:src}
\end{figure*}

\subsection{Experiment Results}

\myparagraph{Maximum regret ratio and CPU time by varying $k$}
In Figure~\ref{fig:MHTk}, we report the maximum regret ratio and CPU time on all the datasets by varying $k$.
We omit the result of the \AlgNWF algorithm on the Weather dataset because it takes several days to obtain a solution.
For the \AlgKER algorithm, the result size of this algorithm is much larger than $k$, so we randomly return $k$ points as the solution as in \cite{Xie:2018}.
In high-dimensional datasets such as Movie and Weather, the \AlgHDR, \AlgHDG and \AlgKER algorithms all require too much memory to compute the solutions. Thus, we omit their performances in the Movie and Weather datasets. When $k<d$, \AlgHDG and \AlgSPH do not work \cite{Asudeh:2017,Xie:2018}, so we report their performance for $k\geq d$.

As shown, 
when $k$ is small, the maximum regret ratios of our proposed algorithms especially \AlgSPG are worse than the others. The reason is that \AlgSPG uses greedy policy and sampling to select points. But with $k$ increase, they are very close to each other except for the NBA and Household datasets. For the NBA dataset, our \AlgSPG performs better than HD-based algorithms and is comparable to others. 
In addition, when we need to select fewer data points in high-dimensional datasets (\ie $k<d$), such as Movie and Weather, the \AlgSPH algorithm does not work under this case, but our \AlgSPG algorithm can still perform well.
\AlgPG and \AlgSPG with different values of $\epsilon$ and $\lambda$ are very close to each other and decrease with $k$. For the real datasets, the maximum regret ratio is much smaller than that on the anti-correlated dataset, and they generally decrease with the values of $k$.

In terms of CPU time, by increasing $k$, the response time of all the algorithms increases since more data points are selected.
\AlgNWF, \AlgHDR and \AlgKER run much slower than the other algorithms on all the datasets. Under different settings of $\epsilon$ and $\lambda$,
\AlgSPG constantly outperforms \AlgPG and achieves up to 10 times speedup. The \AlgHDG algorithm is close to \AlgSPG, but it consumes a large amount of memory. Therefore, we omit it for the Movie and Weather datasets. The performance of \AlgSPG is close to \AlgSPH and even better in the Household dataset. However, \AlgSPH cannot work under the case of $k<d$.
As can be observed, for different values of $\epsilon$ and $\lambda$, the maximum regret ratio
reported by \AlgSPG is very close to that of \AlgPG but with much better performance in CPU time.

In summary, we can observe that \AlgSPG provides very compelling performance and can make a tradeoff between the CPU time and the quality of returned results. Though our \AlgSPG algorithms sacrifices a little regret ratio, it can work under any circumstances, especially for the large high-dimensional datasets, \eg the Weather dataset. 



\myparagraph{Maximum regret ratio and CPU time by varying $d$ and $n$}
For large $d$ or $n$, \AlgNWF and \AlgHDR are too time consuming. Therefore, we omit the evaluation of \AlgNWF and some results of \AlgHDR here.
In Figure~\ref{fig:MHTdn}(\emph{a})-(\emph{b}), we conduct the experiments on the anti-correlated datasets by varying $d$.
When $d$ increases, the maximum regret ratio and CPU time of all algorithms increase since more dimensions need to be processed.
With the increase in $d$, the maximum regret ratios of different algorithms are very close to each other. However, the CPU times of \AlgPG, \AlgKER and \AlgHS are always larger than that of \AlgSPG for different values of $\epsilon$ and $\lambda$. This is because \AlgPG requires more function evaluations, \AlgKER needs more time to compute available coreset and \AlgHS consumes plenty of time to solve a large number of hitting-set problems.
The CPU time of \AlgHDR increases very quickly and the performance drops quickly, especially in high-dimensional spaces.
\AlgHDG performs better for $d<6$, but the query time is longer compared with \AlgSPG and \AlgSPH when $d$ is large. This is because the discretized matrix in \AlgHDG can be large in high-dimensional spaces. The CPU time of \AlgSPG is close to \AlgSPH, but \AlgSPH has the aforementioned drawback when $k<d$.
We vary $n$ on the anti-correlated dataset and the results are shown in Figure~\ref{fig:MHTdn}(\emph{c})-(\emph{d}).
Similar trends can be observed among the algorithms.

\myparagraph{Regret ratio distribution}
Since \AlgSPG is a sampling-based method, we report the regret ratio distribution here to demonstrate the stability of \AlgSPG under different $\epsilon$ and $\lambda$ settings.
We run the algorithms 20 times with $k=25$ and record their distribution on all the datasets. The results are shown in Figure~\ref{fig:MHTdis}.
Figure~\ref{fig:MHTdis}(\emph{a}) only shows the regret ratios of the \AlgPG and \AlgSPH algorithms, which are the algorithms without sampling. The other nonsampling methods are quite time consuming, so we omit them.
Figure~\ref{fig:MHTdis}(\emph{b})-(\emph{e}) shows the regret ratio distribution for \AlgSPG. It records the number of times that the regret ratio has fallen into that interval. The size of each interval is set to approximately 0.05. We can see that most of the regret ratios of \AlgSPG are very close to those of \AlgPG.
The difference in the regret ratio between our \AlgSPG algorithm and the \AlgPG algorithm does not exceed 0.1. For the anti-correlated and Movie datasets, the difference does not exceed 0.05.
In addition, the difference of the regret ratio between our \AlgSPG algorithm and the \AlgSPH algorithm is similar to that of \AlgPG on the anti-correlated and NBA datasets. 
However, due to the high dimensions of the Movie and Weather datasets, \AlgSPH either does not work or performs poorly when computing the $k$ results. Our \AlgSPG algorithm can still provide stable performance.
We do not show the results on the Household dataset here, because the regret ratios of the \AlgPG, \AlgSPG and \AlgSPH algorithms are almost the same. Since the algorithms are conducted over the skyline points, the larger the skyline size is, the more effective our \AlgSPG algorithm will be.

\myparagraph{Submodularity ratio and curvature} We compute the submodularity ratio for each $S_i$, where $i\in[1,k-1]$. The results are shown in Figure~\ref{fig:src}(\emph{a})-(\emph{e}). We can see that the submodularity ratio is nonincreasing with the increase in $k$ for \AlgPG. For the \AlgSPG algorithm, its submodularity ratio is not strictly nonincreasing with the increase in $k$, due to the randomness of \AlgSPG.

As illustrated in Figure \ref{fig:src}(\emph{f})-(\emph{j}), with the increase in $k$ the curvatures are almost equal to $1$ for \AlgPG. According to the definition of curvature, for most cases, when $T\subseteq D$ and $|T|=k$, $\Delta_{S_{i-1}\cup T}(q_i)=0$, which makes the curvature equal to 1. From Theorem \ref{thm:preselection}, we can see that the approximation ratio approaches $1-1/e$.
Thus, it verifies the advantages of the \AlgPG algorithm.

\myparagraph{The evaluation of $\epsilon$ and $\lambda$}
In \AlgSPG, a larger $\epsilon$ or $\lambda$ denotes a smaller sample size. Here, we evaluate the performance of
\AlgSPG by varying $\epsilon$ and $\lambda$ on the anti-correlated dataset. The results are shown in Figure~\ref{fig:el}(\emph{a})-(\emph{b}).
As we can observe, the maximum regret ratio does not change much when we vary $\epsilon$ or $\lambda$ over a wide range.
However, the CPU time decreases considerably when the sample size drops.

\begin{figure}[!htb]
	\centering
	\subfigure[Maximum regret ratio]{%
		\includegraphics[width=0.23\textwidth]{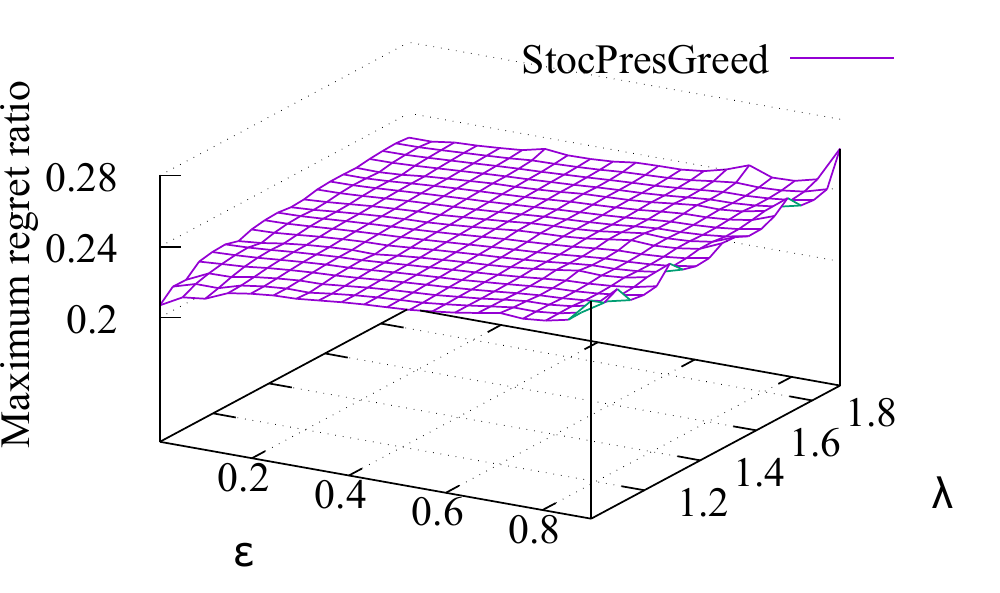}
	}
	\subfigure[CPU time]{%
		\includegraphics[width=0.23\textwidth]{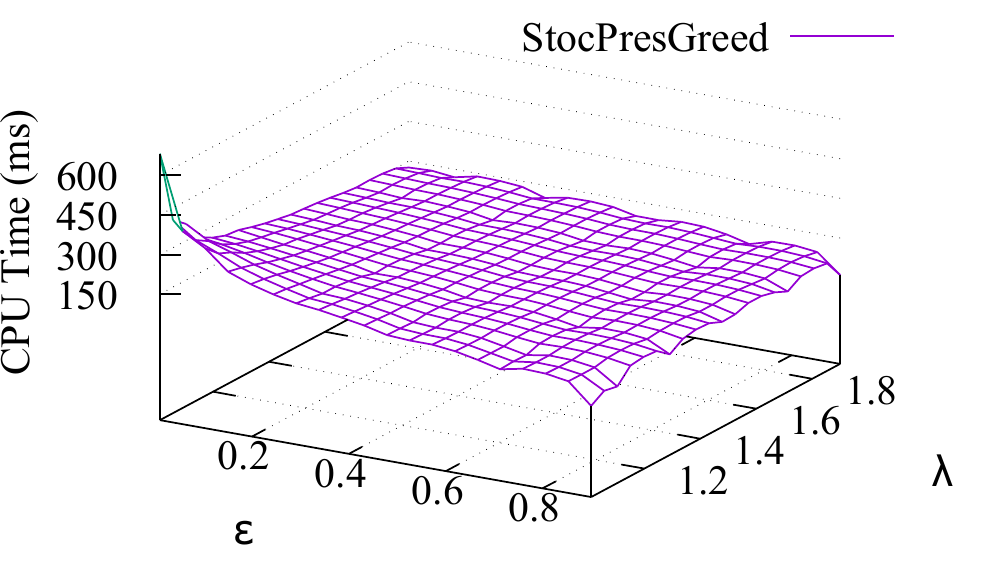}
	}	
	\caption{Performance comparisons. (\emph{a}) and (\emph{b}) show the maximum regret ratios and CPU times of the \AlgSPG for different values of $\epsilon$ and $\lambda$ on the \textit{anti-correlated} datasets.
	}
	\label{fig:el}
\end{figure}

\myparagraph{The upper bounds of maximum regret ratios} Since the \AlgHS, \AlgHDG and \AlgKER algorithms cannot output exact $k$ points, we report ony the upper bounds of the \AlgHDR and \AlgSPH algorithms along with our proposed algorithms. As Figure~\ref{fig:theo}(\emph{a})-(\emph{e}) shows, the upper bound of the maximum regret ratio of \AlgHDR is smaller than others and closer to the maximum regret ratio of the optimal solution. However, \AlgHDR is time consuming as shown in Figure \ref{fig:MHTk} and \ref{fig:MHTdn}. The upper bounds of the maximum regret ratios of \AlgPG and \AlgSPG are close to each other. 
For the \AlgSPH algorithm, its upper bounds are decided by both $k$ and $d$. As shown in Figure~\ref{fig:ubsphere}, the upper bound of the maximum regret ratio increases with increasing $d$, but for $k$, it basically remains unchanged.

\begin{figure*}[!htb]
	\centering	
	\subfigure[Anti-correlated]{\includegraphics[width=0.19\textwidth]{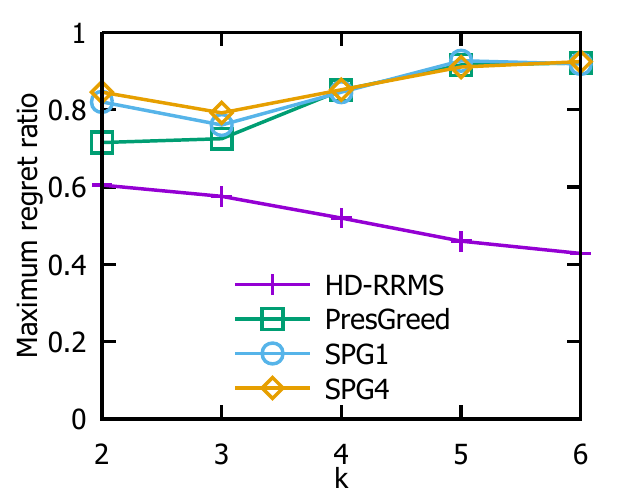}}
	\subfigure[NBA]{\includegraphics[width=0.19\textwidth]{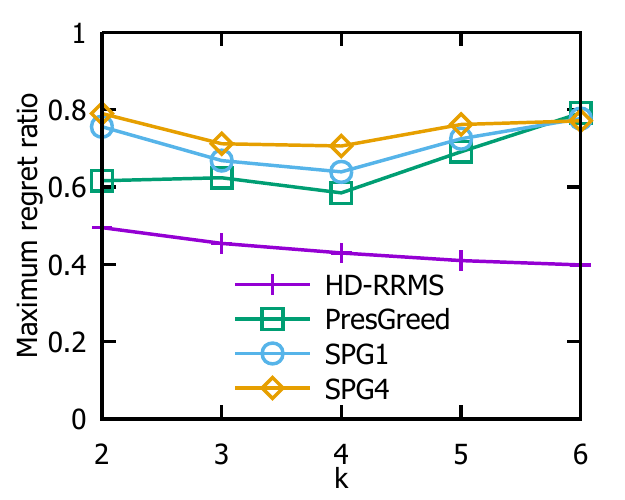}}
	\subfigure[Household]{\includegraphics[width=0.19\textwidth]{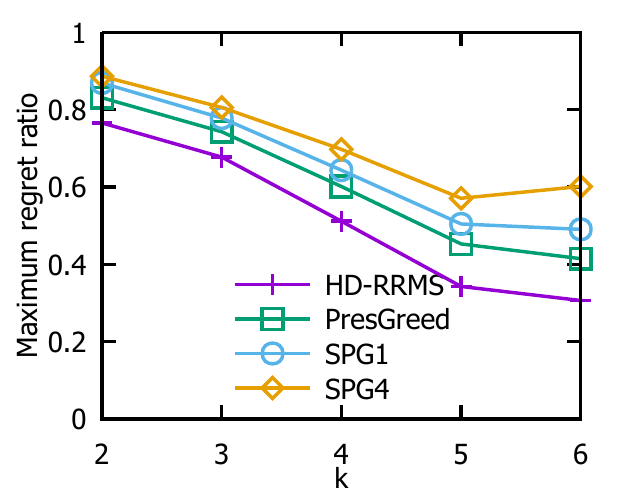}}
	\subfigure[Movie]{\includegraphics[width=0.19\textwidth]{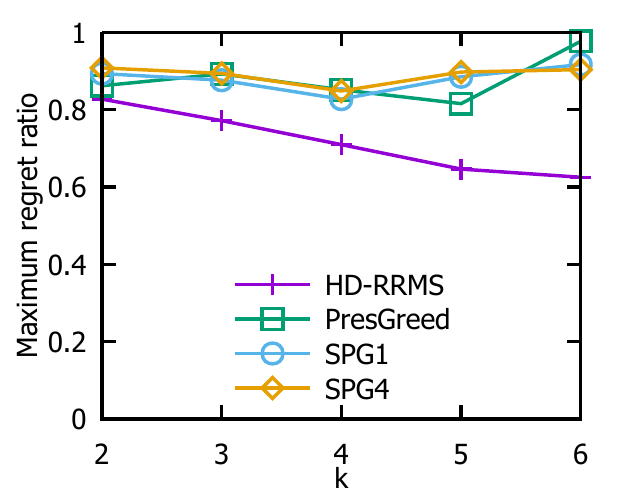}}
	\subfigure[Weather]{\includegraphics[width=0.19\textwidth]{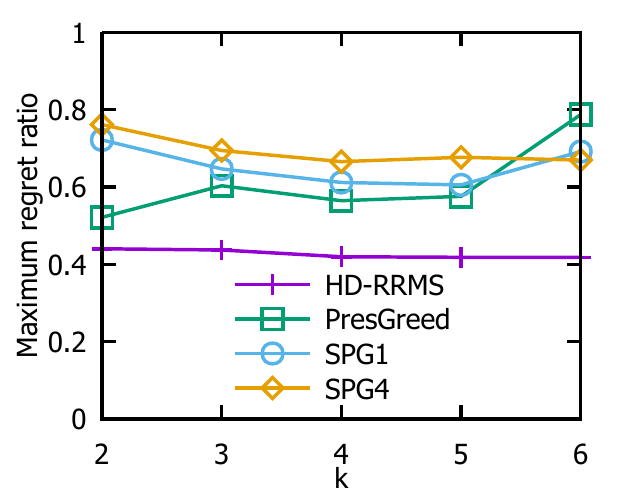}}
	\caption{Performance comparisons. (\emph{a})-(\emph{e}) show the upper bounds of the maximum regret ratio for different values of $k$ on \textit{Anti-correlated}, \textit{NBA}, \textit{Household}, \textit{Movie} and \textit{Weather}, respectively.}
	\label{fig:theo}
\end{figure*}

\begin{figure}[!htb]
	\centering
	\subfigure{\includegraphics[width=0.35\textwidth]{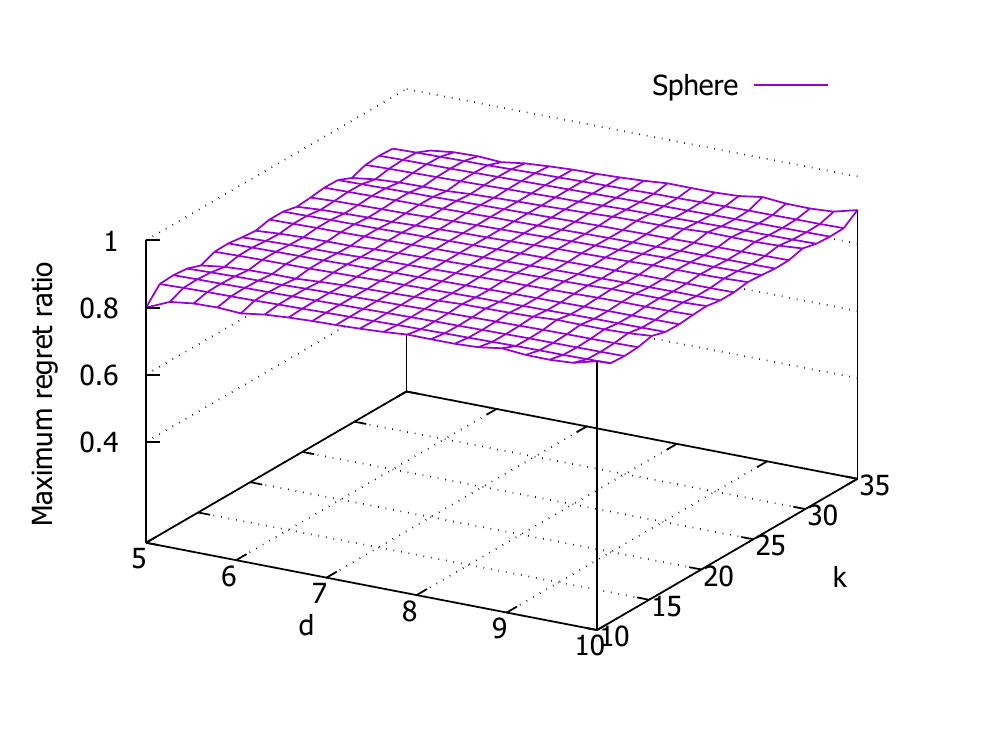}}
	\caption{Upper bounds of the maximum regret ratio of \AlgSPH for different values of $d$ and $k$.}
	\label{fig:ubsphere}
\end{figure}

\section{Related Work}
\label{sec:rel}

Top-$k$ queries \cite{Ilyas:2008} and skyline queries \cite{Borzsony:2001} are two well-known tools for multi-criteria decision making that have received considerable attention during the last two decades. 
Improved versions are also investigated to improve the quality of the returned results, such as the subspace skyline \cite{Yuan:2005}, the constrained skyline \cite{Papadias:2005,Zhang:2010}, the restricted skyline \cite{Ciaccia:2017,Ciaccia:2020}, the $k$-dominant skyline query \cite{Chan:2006}, top-$k$ representative skyline \cite{lin:2007}, distance-based representative skyline \cite{Tao:2009} and threshold-based preferences \cite{Sarma:2011}. 


Due to the inherent limitations of top-$k$ and skyline queries, Nanongkai et al. \cite{Nanongkai:2010} proposed the $k$-regret minimization problem.
The problem is NP-hard, and a greedy strategy is proposed.
However, the developed approach suffers from efficiency issues, and cannot scale well for large datasets. The following studies attempt to solve the problem from different perspectives.
Peng et al.~\cite{Peng:2014} attempted to reduce the candidate points from the whole skyline points to a small candidate set by using geometric properties.
Cao et al. \cite{Cao:2017} and Agarwal et al. \cite{Agarwal:2017} proposed the $\epsilon$-kernel algorithm independently, which can keep the maximum regret ratio at most $\epsilon$ with an output size of $O(\epsilon^{-\frac{d-1}{2}})$. However, the output sizes of these algorithms are uncontrollable. 
They made output size from $O(\epsilon^{-\frac{d-1}{2}})$ to $k$. Then the upper bound was $O(k^{-\frac{2}{d-1}})$ \cite{Xie:2018}. Agarwal et al. \cite{Agarwal:2017} also proposed the hitting-set algorithm, which returns $O(k\log k)$ points.
Asudeh et al. \cite{Asudeh:2017} interpreted the $k$-regret query in a $d$-dimensional dataset as a discretized matrix min-max problem and proposed two algorithms with upper bounds. 
Moreover, Xie et al. \cite{Xie:2018} designed the state-of-the-art Sphere algorithm whose upper bound on the maximum regret ratio was asymptotically optimal and restriction-free for datasets of any dimensionality. These bounds or theoretical results focus on the regret ratio instead of approximation guarantees. Qiu et al. \cite{Qiu:2018} and Dong et al. \cite{Dong:2019} also provided the sampling techniques for 1-RMS and $k$-RMS queries, respectively. However, the proposed method in \cite{Qiu:2018} is a special case of ours, where the sample factor equals 1, \ie $\lambda=1$. Additionally, they both do not provide any approximation guarantees. 

Several works extend Nanongkai et al. \cite{Nanongkai:2010} to some extent. Chester et al. \cite{Chester:2014} introduced the relaxation to $k$-regret minimizing sets, denoted as the $k$-RMS query. The efficiency of the $k$-RMS query was further improved by \cite{Cao:2017} based on the $\epsilon$-kernel and \cite{Agarwal:2017,Kumar:2018} via coresets and hitting sets. Faulkner et al. \cite{Faulkner:2015} and Qi et al. \cite{Qi:2018} extended the linear utility functions to \textsc{Convex}, \textsc{Concave}, \textsc{CES} utility functions and multiplicative utility functions for $k$-regret queries. 
Zeighami and Wong \cite{Zeighami:2019} proposed the metric of the average regret ratio to measure user satisfaction. 
To reduce the bounds of the regret ratio, Nanongkai et al. \cite{Nanongkai:2012} and Xie et al. \cite{Xie:2019} combined user interactions into the process of selection. 
Additionally, based on regret minimization concepts, \cite{Asudeh:2017,Asudeh:2019,Shetiya:2020,Xie:2020} focus on compact maxima, rank regret representative problems, unified algorithms for different aggregate norms and min-size versions of the regret minimization query. 
To reduce the bounds of regret ratio, Nanongkai et al. \cite{Nanongkai:2012} combine user's interactions into the process of selection. Moreover, Xie et al. \cite{Xie:2019} provide a strongly truthful interactive mechanism to leverage the regret ratios using true database tuples instead of artificial ones and provide provable performance guarantees. Zeighami and Wong \cite{Zeighami:2016} propose the metric of average regret ratio to measure user's satisfaction and further develop efficient algorithms to solve it \cite{Zeighami:2019}. Also, based on regret minimization concepts, \cite{Asudeh:2017,Asudeh:2019} focus on the compact maxima and rank regret representative problems, respectively. 
Moreover, the concept of the regret ratio is also adopted to solve the problems in the machine learning area, e.g., multi-objective submodular function maximization \cite{Soma:2017}. 

Due to the hardness of the problem, most of the studies utilize the greedy framework \cite{Peng:2014,Faulkner:2015,Xie:2018,Qi:2018,Qiu:2018,Dong:2019}. However, there is still no strict theoretical analysis of the approximation ratio in the original greedy framework. We fill the theoretical gap, and our developed techniques can be used by the existing research if the greedy framework is used.

For summodularity aspect, submodular function maximization has numerous applications in machine learning and database systems \cite{Soma:2014,Lin:2011}. Though the problems are NP-hard,
the greedy method can return a result with $1-1/e$ approximation ratio~\cite{Nemhauser:1978}. By leveraging the marginal gain property, lazy-update method is proposed to accelerate the search~\cite{Minoux:1978}. Badanidiyuru and Vondr\'{a}k \cite{Badanidiyuru:2014SODA} propose a centralized algorithm that achieves a $(1-1/e-\epsilon)$ approximation ratio using $O(n/\epsilon \log(n/\epsilon))$ function evaluations for general submodular functions, and a multistage algorithm is proposed to further accelerate the search.
In \cite{Mirzasoleiman:2015}, a linear time algorithm is proposed for cardinality constrained submodular maximization, which provides the same approximation ratio as \cite{Badanidiyuru:2014SODA} with $n\log \frac{1}{\epsilon}$ function evaluations.
The accelerated greedy algorithms \cite{Minoux:1978,Mirzasoleiman:2015} start from an empty set, and need more time for picking the first data point. While the approaches in \cite{Badanidiyuru:2014SODA} require much more functions evaluations when $n$ is large, and they are more effective when those submodular functions can be easily decomposed and approximated. In a word, all the approaches in above are not efficient for large-scale datasets.

\section{Conclusion}
\label{sec:conc}

As an important operator, the $k$-regret minimization query is investigated to provide users with high-quality results from a large dataset. The existing solution is efficient, but there is still much space left for improvement, especially for larger datasets. In addition, the approximation ratio of the returned results of the greedy framework is not discussed in the literature. In this paper, we conduct the first theoretical analysis and provide an approximation guarantee for the previous \AlgRDP greedy framework by utilizing the proposed happiness ratio concept. To reduce the evaluation cost and further speed
up the processing, a sampling-based method, \AlgSPG, is proposed that provides an $(1-e^{-\frac{(1-\epsilon)(k-1)\gamma}{\lambda k}})$ approximation ratio.
Moreover, careful analysis is presented to demonstrate the sample size required.
Experiments over real-world and synthetic datasets were conducted to verify the advantages of the proposed methods.

\section*{Acknowledgment}
This work is partially supported by the National Natural Science Foundation of China under grants U1733112, 61702260, 61802345 and the Fundamental Research Funds for the Central Universities under grant NS2020068.

\bibliographystyle{IEEEtran}
\bibliography{ref}

\end{document}